\documentclass[preprint,11pt,authoryear]{elsarticle}
\usepackage{tikz}
\usepackage{xcolor}
\usepackage{amsfonts}
\usepackage{amsmath}
\usepackage{amsthm}
\usepackage{amssymb}
\usepackage{times}
\usepackage{subfig}
\usepackage{multirow}
\usepackage{setspace}
\usepackage{enumerate}
\usepackage{natbib}
\usepackage{color}
\usepackage{colortbl}
\usepackage{ wasysym }
\usepackage{verbatim}
\usetikzlibrary{shapes,arrows,calc,positioning}
 \usepackage{lscape}
\usepackage{array}
\usepackage{setspace}
\usepackage{url}
\usepackage{arydshln}

\usepackage{caption}
\usepackage{rotating}
\newcolumntype{x}[1]{%
>{\centering\hspace{0pt}}p{#1}}%

\definecolor{Gray}{gray}{0.9}
\newcolumntype{g}{>{\columncolor{Gray}}c}

\def\iid{\buildrel {\rm i.i.d.} \over \sim}

\def\i.i.d.{\buildrel {\rm i.i.d.} \over \sim}

\def\cw#1 { \overset{\mathbb{P}}{\underset{#1}{\longrightarrow}} }
\def\Real{\mathbb{R}}
\def\Natu0{\mathbb{N}_0}
\def\P#1{{\mathbb{P}}\left(#1\right)}

\def\E#1{{\mathbb E}\left[#1\right]}

\def\Var#1{{\mathrm Var}\left(#1\right)}

\def \rcov#1#2 {{\rm cov}_{#1}\left( #2\right)}

\DeclareMathOperator*{\argmin}{arg\,min}

\newtheorem{example}{Example}

\newtheorem{algo}{Algorithm}

\newtheorem{lemma}{Lemma}
\newtheorem{theorem}{Theorem}
\newtheorem{definition}{Definition}
\newtheorem{corollary}{Corollary}
\newtheorem{remark}{Remark}

\newtheorem*{toy*}{Toy Model}

\newtheorem{criteria}{Criteria}
\newtheorem{BMS}{Bonus-Malus System}

\definecolor{Gray}{gray}{0.9}

\begin{document}
\begin{frontmatter}
\title{Double-Counting Problem of the Bonus-Malus System}%{On the Efficiency of Bonus-Malus System}

%\author[IH]{Woojoo Lee}
%\ead{lwj221@gmail.com}
%\address[IH]{Department of Statistics, Inha University, 235 Yonghyun-Dong, Nam-Gu, Incheon 402-751, Korea.}
%\author[BSNU]{Sojung C. Park}
%\ead{sojungpark@snu.ac.kr}
%\address[BSNU]{College of Business Administration, Seoul National University, Gwanak-gu, Seoul 151-916, Republic of Korea.}

\author[EHM]{Rosy Oh\corref{cor1}}
\ead{rosy.oh5@gmail.com}
\address[EHM]{Institute of Mathematical Sciences, Ewha Womans University, Seodaemun-Gu, Seoul, Republic of Korea}

\author[KS]{Kyung Suk Lee\corref{cor1}}
\ead{leeks@kongju.ac.kr}
\address[KS]{Department of Physics Education, Kongju National University, Gongju, Republic of Korea}

\author[BSNU]{Sojung C. Park\corref{cor1}}
\ead{sojungpark@snu.ac.kr}
\address[BSNU]{College of Business Administration, Seoul National University, Gwanak-gu, Seoul, Republic of Korea}

%\author[WI]{Peng Shi\corref{cor2}}
%\ead{pshi@bus.wisc.edu}

\author[EH]{Jae Youn Ahn\corref{cor2}}
\ead{jaeyahn@ewha.ac.kr}
\address[EH]{Department of Statistics, Ewha Womans University, Seodaemun-Gu, Seoul, Republic of Korea}
%\address[WI]{Wisconsin School of Business, University of Wisconsin-Madison, Madison, WI 53706, USA.}
\cortext[cor1]{First Authors}
\cortext[cor2]{Corresponding Authors}

\begin{abstract}
The bonus-malus system (BMS) is a widely used premium adjustment mechanism based on policyholder's claim history. Most auto insurance BMSs assume that policyholders in the same bonus-malus (BM) level share the same {\it a posteriori} risk adjustment. This system reflects the policyholder's claim history in a relatively simple manner. However, the typical system follows a single BM scale and is known to suffer from the double-counting problem: policyholders in the high-risk classes in terms of {\it a priori} characteristics are penalized too severely \citep{Taylor1997, Pitrebois2003}. Thus, \citet{Pitrebois2003} proposed a new system with multiple BM scales based on the {\it a priori} characteristics. While this multiple-scale BMS removes the double-counting problem, it loses the prime benefit of simplicity.
Alternatively, we argue that the double-counting problem can be viewed as an inefficiency of the optimization process. Furthermore,
we show that the double-counting problem can be resolved by fully optimizing the BMS setting, but retaining the traditional BMS format.

\end{abstract}


\end{frontmatter}

\vfill

\pagebreak

\vfill

\pagebreak

%\section*{Double-Counting Problem of the Bonus-Malus System}%{On the Efficiency of Bonus-Malus System}

\section{Introduction}

The bonus-malus system (BMS), also called the merit rating system, or no-claim bonus system, is a widely used experience rating system in auto insurance. The system corrects {\it a priori} risk classification by reflecting the policyholder's claim history {\it a posteriori}. In a typical auto insurance rating system, the rate is first determined by the observable characteristics of the policyholder {\it a priori}, and then adjusted under the BMS from the policyholder's unobservable information embedded in the claim history.
Bonus-malus (BM) is a very important rating factor. \citet{Lemaire1998} states that ``if insurers are allowed to use only one rating variable, it should be some form of merit-rating," and this is indeed true in practice.

The design and implementation of the BMS vary significantly across countries and insurers, involving from nonexistent to highly complex systems. Highly complex systems may achieve the goal of accurately predicting future claims but might increase the computational burden of insurers and make their communication with policyholders and regulators difficult. They may also challenge the effectiveness of supervision and increase consumer complaints. This study suggests that the BMS might enable a simple presentation to insurance consumers and still be as accurate as a very complex system by resolving the so-called double-counting problem of the traditional systems.

A traditional Asian-European auto insurance BMS (the traditional BMS hereafter) consists of finite discrete risk levels, rules for transition between the levels, and {\it BM relativity} assigned to each level. Policyholders can navigate between the risk levels by following the transition rules based solely on their claim history. A single BM relativity or scale assigned to each level determines the magnitude of the bonus (discount) or malus (penalty). This BM relativity revises the rate based on {\it a priori} risk classification (see \citet{LemaireZi1994} for numerous examples).

This simple yet effective system is found to have a challenging double-counting problem due to a technical bias in the traditional BMS from the {\it a priori} characteristics such as age and gender affecting both the {\it a priori} and {\it a posteriori} risk classifications. For example, a young male driver is first penalized by the age and gender factors due to the high probability of having claims, and then again penalized from the realized claim records and the BMS.
This critical problem was first raised by \citet{Taylor1997} in their seminal paper, and then well presented by \citet{Brouhns2003, Pitrebois2003, Lemaire2015}. In an attempt to resolve the double-counting problem, \citet{Pitrebois2003} proposed a fair system, where each {\it a priori} class has its own BMS (individualized BMS hereafter), through a simple example, and \citet{Brouhns2003} suggested a computer-intensive individualized BMS method incorporating the interactions between the {\it a priori} and {\it a posteriori} mechanisms.

Since the 1990s, several countries have deregulated their BMS, moving away from a strict national system in which all insurers use the same levels, transition rules, and BM relativities. European and Asian insurers began to use intermediate or total freedom systems, designing their own systems that diverged significantly from the traditional systems \citep{Lemaire1998}.
Considering the effectiveness and significance of the BMS for auto insurance rates, deregulation of the BMS can be viewed as promoting price competition among insurers. In addition, deregulation of the BMS enables insurers to utilize more complex systems, such as those suggested by \citet{Brouhns2003, Pitrebois2003}, and thus resolve the double-counting problem affecting actuarial equity.

However, this also indicates the transition from simple to more complex and opaque systems. Most insurers do not have a clear transition rule or relativities as seen in simple national systems. Therefore, in countries with a total freedom system, the only way for consumers to know the impact of new claim(s) on their insurance rates is from the quotes of insurers after the claim(s) are settled.
Some regulators and consumers have found the lack of transparency in the BMS affecting consumer protection, with increased complexity making the supervision of rates a challenging task. For example, the New York Insurance Law section 2334 states that it ''shall continue to encourage competition among insurers, but shall discourage merit rating plan provisions which may tend to create confusion or misunderstanding among insureds. Section 2334 further requires insurers to clearly present the surcharge due to claim history in the declarations page or premium bill for transparency. The European Consumer Organization BEUC asks all insurers to mandate the BMS, requiring them to present their BM policies in a standardized template in order to increase policy transparency and comparability. Furthermore, in spite of the double-counting problem, several Asian countries still follow the national system in the traditional BMS format, possibly for socio-economic or practical reasons.

This paper presents a simple BMS system that reduces the double-counting problem and yet maintains its simple traditional format. We propose a new optimization process to set the {\it
a priori} rate and a BM relativity attached to each BM level after revisiting the classical BMS considered in \citep{Taylor1997, Pitrebois2003, Lemaire2015, Chong}. Thus, we have a new premium that can be considered appropriate for the classical BMS model, but with modified numerical values for the {\it a priori} rate and BM relativities.
The main difference of the proposed premium is that the {\it a priori} rate contributes to the optimization process as well as BM relativities. In the classical premium, only BM relativities contribute to the optimization process, and the {\it a priori} rate is set to a fixed value based on a statistical estimation procedure.
The new premium's immediate result is better prediction of future premiums compared to the classical premium. Surprisingly, we find that the new premium hardly faces the double-counting problem. We explain that the double-counting problem in the classical premium arises from inefficiency in the optimization process, where the {\it a priori} rate is not allowed to participate in the optimization process.% is now allowed to participate the optimization process.

The remainder of this paper is organized as follows.
Section 2 reviews the frequency random effects model and the classical BMS.
Section 3 investigates the double-counting problem in the classical BMS.
Specifically, we formally define the double-counting
problem mathematically and propose an index to quantify the problem.
In Section 4, we propose a new premium under the full optimization process.
We also explain the double-counting problem in the classical BMS as mainly due to the constrained {\it a priori} rate,
which we allow to participate in the new optimization process in this study.
Using a real data analysis in Section 5, we show that the double-counting problem in the classical premium can be resolved by adjusting the {\it a priori} rate and BM relativities.

\section{Review of the Classical BMS}

\subsection{Definition and Notations}

 We consider the policyholders' portfolio in the short-term insurance context, where the policyholders can decide whether or not to renew their policy at the end of each policy year, and insurers can adjust their premium at the beginning of each policy year based on the policyholders' claim history. We denote $\mathbb{N}$, $\mathbb{N}_0$, $\Real$, and $\Real^+$ as the sets of natural numbers, non-negative integers, real numbers, and positive real numbers, respectively.
We define $N_{i,t}$ to indicate the $i$-th policyholder's number of claims in the $t$-th policy year.
The actuarial science literature often refers to $N_{i,t}$ as the insurance claim \textit{frequency}. Furthermore, we define the $i$-th policyholder's claim history at the end of year $t$ as
\[
\mathcal{F}_{i,t}:=\left\{n_{i1} \cdots n_{i,t}\right\}.
\]

Let $\boldsymbol{X}_i$ be the {\it a priori} risk characteristics of the $i$-th policyholder observable at the time of contract, and $\Theta_i$ be the residual effect, or the policyholder's characteristics not included in the {\it a priori} risk classification.
In the BMS context, we distinguish between the \textit{a priori} and \textit{a posteriori} risk classification. The {\it a priori} rate function is based solely on the policyholders' {\it a priori} risk characteristics, without considering their claim history $\mathcal{F}_{i,t}$, and the {\it a posteriori} risk classification is related to the policyholders' residual effect $\Theta_i$ based on their claim history $\mathcal{F}_{i,t}$.

Following the classical settings in the BMS,
we assume constant {\it a priori} risk characteristics and residual effects, making it convenient to analyze the stationary frequency distribution.
Thus, for $\boldsymbol{X}_i$ and $\Theta_i$, we do not have the subscript $t$.
To model the insurance claim frequencies, we use the generalized linear model (GLM) format applied to Poisson distribution.
Nonetheless, this paper's arguments can apply to any counting distribution, including
the Poisson and negative binomial distributions as well as the distributions for excess zeros, such as the zero-inflated or hurdle models \citep{Nelder1989, yip2005modeling, deJong2008}.

\subsection{Frequency Random Effects Model}\label{sec.3}

This section presents the general frequency modeling framework. Assume that $\boldsymbol{x}_{\kappa}$ defines the {\it a priori} risk characteristics of the $\kappa$-th risk class, and $w_\kappa$ is the weight of the corresponding risk class; that is,
  \begin{equation}\label{eq.1}
  w_\kappa:= \P{\boldsymbol{X}_{i}=\boldsymbol{x}_{\kappa}}, \quad \kappa = 1,\cdots, \mathcal{K}.
  \end{equation}
Denote $\Lambda_i$ as the {\it a priori} rate for the $i$-th policyholder, determined by the {\it a priori} risk characteristics based on equation
  \begin{equation}\label{eq.2}
        \Lambda_{i}=\eta^{-1} \left(\boldsymbol{X}_{i}\boldsymbol{\beta}\right),
  \end{equation}
where $\eta(\cdot):\Real \mapsto \Real_+$ is the link functions and $\boldsymbol{\beta}$ is the vector of parameters to be estimated. The residual effect $\Theta_i$, assumed to be independent of the {\it a priori} risk characteristics
$\boldsymbol{X}_{i}$, has the distribution $G$:
   \begin{equation}\label{eq.4}
  \Theta_i \iid G, \quad i=1, \cdots, I,
  \end{equation}
where $G$ denotes the continuous marginal distribution function for $\Theta$. We use $g$ to denote the density function corresponding to $G$.
Further, for identification purposes, we assume that
   \begin{equation}\label{eq.3}
       \E{\Theta_i} = 1.
  \end{equation}
Now, we are ready to present the frequency random effect model, which is assumed throughout the study unless specified otherwise.
\bigskip

\noindent
{{\bf Frequency Random Effect Model.} \it
Consider the {\it a priori} risk characteristics described in \eqref{eq.1} and \eqref{eq.2}, and the residual effect in \eqref{eq.4} and \eqref{eq.3}. We assume that the {\it a priori} risk characteristics and the residual effect are independent.
Assuming the {\it a priori} risk characteristics $\boldsymbol{X}_{i}=\boldsymbol{x}_{\kappa}$
and residual effect $\Theta_i=\theta_i$,
we construct the frequency model for $N_{i,t}$ using the count regression model
       \begin{equation}\label{eq.n}
        N_{i,t}\big\vert \left(\boldsymbol{X}_{i}, \Theta_i\right)=\left(\boldsymbol{x}_{\kappa}, \theta_i\right) \iid { F}(\lambda_{\kappa}\theta_i),
        \end{equation}
where the non-negative integer distribution function $F$ has the mean parameter of $\lambda_{\kappa}\theta_i$ with
   $\lambda_{\kappa}= \eta^{-1}\left(\boldsymbol{x}_{\kappa}\boldsymbol{\beta}\right)$.
Clearly, the regression model in \eqref{eq.n} can be similarly written as
          \begin{equation*}%\label{eq.n2}
        N_{i,t}\big\vert \left(\Lambda_{i}, \Theta_i\right)=\left(\lambda_{\kappa}, \theta_i\right) \iid
        {F}(\lambda_{\kappa}\theta_i),
        \end{equation*}
   for the random variable $\Lambda_i$ independent of $\Theta_i$ and satisfying $\P{\Lambda_i=\lambda_\kappa}=w_\kappa$.
}
\bigskip

 The mean of $N_{i, t+1}$ for the given claim history up to time $t$, denoted by $M_{\rm [Bayes]}(\lambda_\kappa, \mathcal{F}_{i,t})$, is the {\it Bayesian premium} in an insurance setting. The formal mathematical definition of the Bayesian premium is
\begin{equation}\label{eq.6}
\begin{aligned}
M_{\rm [Bayes]}\left(\lambda_\kappa, \mathcal{F}_{i,t}\right)&:=\E{N_{i,t+1}\big\vert \Lambda_i=\lambda_{\kappa}, \mathcal{F}_{i,t}} \\
&= \E{\E{N_{i,t+1}\big\vert \Theta_i, \Lambda_i=\lambda_\kappa,  \mathcal{F}_{i,t}}\big\vert \Lambda_i=\lambda_\kappa, \mathcal{F}_{i,t}}\\
&=\lambda_\kappa\E{\Theta_{i}\big\vert \Lambda_i=\lambda_\kappa, \mathcal{F}_{i,t}},
\end{aligned}
\end{equation}
where the {\it Bayesian estimator of the residual effect}, $\E{\Theta_i\big\vert  \Lambda_i=\lambda_\kappa, \mathcal{F}_{i,t}}$, can be interpreted as the residual effect estimator $\Theta_i$ for the given {\it a priori} rate $\lambda_\kappa$ and claim history up to time $t$.
Thus, we can say that the Bayesian premium is {\it unbiased} in that
\begin{equation}\label{eq.41}
\begin{aligned}
\E{M_{\rm [Bayes]}\left(\Lambda_i, \mathcal{F}_{i,t}\right)\big\vert \Lambda_i=\lambda_\kappa}
&=\lambda_\kappa.\\
\end{aligned}
\end{equation}
Similarly, we can say that the Bayesian estimator of the residual effect is {\it unbiased} if the estimator mean conditioned on the {\it a priori} risk characteristics is 1:
\begin{equation}\label{eq.42}
\E{\E{\Theta_i\big\vert \Lambda_i, \mathcal{F}_{i,t}}\big\vert \Lambda_i=\lambda_\kappa} = 1.
\end{equation}
For brevity and clarity, we omit the subscript $i$ in the remaining part of the study.

\subsection{BMS and Optimal Relativity}\label{OR}

 In the classical BMS, the $-1/+h$ transition rule is common, where
one BM level is lowered for a claim-free year and $h$ BM levels are increased for a claim. We assume the $-1/+h$ transition rule in the remainder of this study unless specified otherwise.\footnote{Note that our assumption of the $-1/+h$ transition rule is only for expository convenience, with the main idea of this study valid for the general transition as long as it is based on the Markovian claim history.} From this transition rule, the BM level for each policyholder would evolve in a Markov chain. Especially, we use $L$ as a random variable representing the BM level (from $1$ to $z$) for a randomly chosen policyholder in the stationary state. The stationary distribution of BM level is determined as
\[
\P{L=\ell}
 =\sum\limits_{\kappa \in\mathcal{K}} w_{\kappa} \int \pi_\ell\left(\lambda_\kappa \theta \right) g(\theta){\rm d}\theta, \qquad \hbox{for} ~~\ell=1, \ldots, {z},
\]
where $\pi_\ell\left(\lambda_\kappa\,\theta \right)$ is the stationary distribution for the policyholder, with frequency $\lambda_\kappa\,\theta$ expected in level $\ell$.
The relativity associated with this BM level $\ell$ is denoted as ${\gamma}(\ell)$, and is called BM relativity.
Under this BMS, the premium denoted as $M(\lambda_\kappa, \ell)$ is a function of the {\it a priori} rate $\lambda_\kappa$ and BM level $\ell$. Throughout this study, we present various types of premium $M(\lambda_\kappa, \ell)$.

From among the various types of BMSs, we review the classical BMS, where all the policyholders share the same BM relativity table
\begin{equation}\label{eq.10}
\gamma(1)=\gamma_{[{\rm Shared}]}(1), \cdots, \gamma(z)=\gamma_{[{\rm Shared}]}(z),
\end{equation}
and the premium is determined as
\begin{equation*}%\label{eq.5}
  M_{[{\rm Shared}]}(\lambda_\kappa, \ell):=\xi_{[{\rm Shared}]}(\lambda_\kappa) \gamma_{[{\rm Shared}]}(\ell),
\end{equation*}
where $\xi_{[{\rm Shared}]}$ is the {\it a priori} rate function determined from the policyholder's {\it a priori} risk characteristics $\boldsymbol{x}_\kappa$.
This BMS, in which all the policyholders share the relativity table in \eqref{eq.10},
is called the {\it {\it BMS with shared BM relativity table}}. Because of its simplicity, the {\it BMS with shared BM relativity table} is a common experience ratemaking system in auto insurance. We formally describe this BMS as follows.

\begin{BMS}[{\it BMS with shared BM relativity table}]\label{BMS.1}
  Consider the BMS where the policyholder's premium at the BM level $\ell$ and {\it a priori} rate $\Lambda=\lambda_\kappa$ is given as
  \begin{equation}\label{eq.5}
  M_{[{\rm Shared}]}(\lambda_\kappa, \ell):=\xi_{[{\rm Shared}]}(\lambda_\kappa) \gamma_{[{\rm Shared}]}(\ell).
\end{equation}
  This BMS is with shared BM relativities.
\end{BMS}

The canonical choice of the {\it a priori} rate function in the {\it BMS with shared BM relativity table} in the literature \citep{Norberg, Denuit2,Lemaire2,Chong} is
\begin{equation*}
\tilde{\xi}_{[{\rm Shared.P}]}(\boldsymbol{x}_\kappa):=\eta^{-1}\left( \boldsymbol{x}_\kappa\boldsymbol{\beta}\right)=\lambda_\kappa,
\end{equation*}
or equivalently
\begin{equation}\label{eq.5.1}
\tilde{\xi}_{[{\rm Shared.P}]}(\lambda_\kappa):=\lambda_\kappa. \end{equation}
The (optimal) BM relativity is determined from the various optimization settings.
For the various optimal BM relativity versions, we refer to \citet{Norberg, Pitrebois2003, Chong}.
From among the various optimization settings, we prefer the premium minimizing the objective function %criteria
\begin{equation}\label{eq.important}
\E{\left(\Lambda\Theta - M(\Lambda, {L})\right)^2}
\end{equation}
as in \citet{Chong}, since one of the main tasks of the BMS is to estimate the predictive mean of $N_{t+1}$.
We formally present this in Criteria \ref{criteria.1}.
 For the given premium $\tilde{M}(\Lambda, {L})$, the mean square error in \eqref{eq.important} is the {\it hypothetical mean square error} (HMSE), denoted by
\[
{\rm HMSE}\left( \tilde{M}(\Lambda, {L})\right):=\E{\left(\Lambda\Theta - \tilde{M}(\Lambda, {L})\right)^2}.
\]

\begin{criteria}\label{criteria.1}
Consider the mean square error in \eqref{eq.important}.
  Then, from among the various premium choices $M(\lambda_\kappa, \ell)$ in the BMS, we prefer the premium with the smallest mean square error in \eqref{eq.important}.

\end{criteria}
Under Criteria \ref{criteria.1}, \citet{Chong} obtain the optimal BM relativity
\begin{equation}\label{eq.ahn130}
(\tilde{\gamma}_{[{\rm Shared.P}]}(1), \ldots, \tilde{\gamma}_{[{\rm Shared.P}]}(z)):=\argmin_{\gamma} \mathbb{E}[\left(\Lambda\Theta-\tilde{\xi}_{[{\rm Shared.P}]}(\Lambda){\gamma}(L)\right)^2],
\end{equation}
where the right-hand side (RHS) of the non-negative function $\gamma$ is optimized, to obtain the solution
\begin{equation*}%\label{eq.60}
\begin{aligned}
\tilde{\gamma}_{[{\rm Shared.P}]}(\ell)&=\frac{\E{ \Lambda^2 \Theta \big\vert L=\ell}}{\E{\Lambda^2\big\vert L=\ell}}\\ %, \quad \ell=1, \ldots, {z}.
&=\frac{ \sum_{\kappa \in \mathcal{K}} w_{\kappa} \lambda_{\kappa}^2 \int   \theta\, \pi_\ell(\lambda_{\kappa} \theta) g(\theta){\rm d}\theta}
{\sum_{\kappa \in \mathcal{K}} w_{\kappa}\lambda_{\kappa}^2 \int \pi_\ell(\lambda_{\kappa}\theta)g(\theta){\rm d}\theta}.
\end{aligned}
\end{equation*}

Now, we can provide the specific premium method classified in the {\it BMS with shared BM relativity table}. The premiums are commonly utilized in the BMS literature.

\bigskip
\noindent
{\bf Premium from Partial Optimization of Shared BM Relativity Table (PPOS).}
From among the various type of premiums in the {\it BMS with shared BM relativity table}, we consider the premium at the BM level $\ell$ with the {\it a priori} rate $\lambda_\kappa$, given as
\begin{equation*}%\label{eq.n3}
\tilde{M}_{[{\rm Shared.P}]}(\lambda_\kappa, \ell) = \tilde{\xi}_{[{\rm Shared.P}]}(\lambda_\kappa) \tilde{\gamma}_{[{\rm Shared.P}]}(\ell),
\end{equation*}
where the {\it a priori} rate function $\tilde{\xi}_{[{\rm Shared.P}]}$ is determined before optimization, as in \eqref{eq.5.1}, while
the BM relativity function $\tilde{\gamma}_{[{\rm Shared.P}]}$ is determined after optimization, as in \eqref{eq.ahn130}.
This premium is the {\it PPOS},
in that only the BM relativity table is optimized while the {\it a priori} rate function is predetermined.
\bigskip

The PPOS is the same as the premium in \citet{Chong}.
The premiums in the classical BMS \citep{Norberg, Pitrebois2003} can be regarded as the PPOS with some variation in specification of the optimization process. Therefore, without loss of generality, the PPOS can be considered as the premium under the classical BMS.
As a base premium method, we can conveniently utilize the following premium with no {\it a posteriori} risk classification (PNO).

\bigskip
\noindent
{\bf Premium with No a Posteriori Risk Classification (PNO).}
Consider the policyholder's premium at the BM level $\ell$ with the {\it a priori} rate $\lambda_\kappa$ given as
\begin{equation*}%\label{eq.n30}
\begin{aligned}
\tilde{M}_{[{\rm Shared.No}]}(\lambda_\kappa, \ell) &:= \tilde{\xi}_{[{\rm Shared.No}]}(\lambda_\kappa)
\tilde{\gamma}_{[{\rm Shared.No}]}(\ell)\\
&= \lambda_\kappa,
\end{aligned}
\end{equation*}
where the {\it a priori} and {\it a posteriori} rate functions are given as
\[
\tilde{\xi}_{[{\rm Shared.No}]}\left( \lambda_{\kappa}\right)=\lambda_{\kappa}
\quad\hbox{and}\quad
\tilde{\gamma}_{[{\rm Shared.No}]}(\ell)=1,
\]
respectively.
This premium is the {\it PNO}. Note that the PNO can be classified as
the {\it BMS with shared BM relativity table}, where the BM relativity function is given by the constant function of $1$.
\bigskip

From the definition of PPOS, compared to PNO, the PPOS provides a better prediction in that
  \begin{equation}\label{eq.p1}
    \E{\left(\Lambda\Theta- \tilde{M}_{\rm [Shared.No]}(\Lambda,L) \right)^2}
   \le
  \E{\left(\Lambda\Theta- \tilde{M}_{[{\rm Shared.P}]}(\Lambda,L) \right)^2}.
  \end{equation}

\section{Double-Counting Problem in the BMS with Shared BM Relativities}\label{sec.33}

From our investigation, the most common premium type in the classical BMS literature is that from the {\it BMS with shared relativity table}, especially with the {\it a priori} rate function determined as
\[
\xi\left(\lambda_\kappa\right)=\lambda_\kappa.
\]
While such premiums are convenient and easy to understand, they are unfair
in that the policyholders with bad/good {\it a priori} risk characteristics pay more/less than the actual premium.

This unfairness in tariff is due to the double-counting problem, and is studied
in the literature in various contexts \citep{Taylor1997, Lemaire2015, Chong}.
As briefly discussed in \citet{Pitrebois2003}, the double-counting problem can be solved with the {\it BMS with individualized BM relativity table}, where policyholders have their own BM relativity table depending on their {\it a priori} risk characteristics. However, the {\it BMSs with individualized BM relativity table} are not widely accepted in practice, mainly because the insurer does not want to complicate the BMS.

 Motivated from the literature, we provide the formal mathematical definition of
the double-counting problem in this section, and
propose an index that can detect the unfairness of
a given BMS. In other words, we present the double-counting problem index.
With these mathematical tools, we can confirm that
the {\it BMS with individualized BM relativity tables} can fully solve the double-counting problem.

\subsection{Double-Counting Problem}

First, we define the unbiasedness of experience ratemaking in the BMS.
As with the case of $M_{\rm [Bayes]}(\lambda_\kappa, \mathcal{F}_{t})$ in \eqref{eq.6}, the unbiasedness of
premium $M(\Lambda, L)$ in the BMS is defined as follows.

\begin{definition}
Under the BMS with the $-1/+h$ transition rule,
the premium $M(\Lambda, L)$
is {\it unbiased} if it satisfies
\begin{equation}\label{eq.8}
\E{M(\Lambda, L)\big\vert \Lambda=\lambda_\kappa}= \lambda_\kappa.
\end{equation}
Otherwise, the premium $M(\lambda_\kappa, \ell)$ is biased.
\end{definition}

Now, we explain the double-counting problem in the classical BMS, that is, the PPOS. First, recall that the PPOS $\tilde{M}_{\rm [Shared.P]}(\lambda_\kappa, \ell)$% in {\it BMS with shared BM relativity table}
satisfies
\[
\tilde{\xi}_{\rm [Shared.P]}(\lambda_\kappa)=\lambda_\kappa.
\]
Hence, the unbiased condition in \eqref{eq.8} is equivalent to
\begin{equation}\label{eq.n16}
\E{\tilde{\gamma}_{\rm [Shared.P]}(L)\vert \Lambda}=1.
\end{equation}

We can show that \eqref{eq.n16} is not true for the PPOS.
Consider two policyholders $i$ and $j$ with the {\it a priori} rates $\Lambda_i=\lambda_{\rm low}$ and $\Lambda_j=\lambda_{\rm high}$, respectively, where $\lambda_{\rm low}<\lambda_{\rm high}$.
While their residual effects $\Theta_i$ and $\Theta_j$ are unknown, the $j$-th policyholder
tends to have a higher mean frequency $\lambda_{\rm high}\,\Theta_j$ than the mean frequency $\lambda_{\rm low}\,\Theta_i$ of the $i$-th policyholder on average.
Thus, the $j$-th policyholder on average shows higher BM levels than the $i$-th policyholder. Consequently, the policyholder with higher {\it a priori} rate $\lambda_{\rm high}$ tends to have a higher BM relativity on average,
\begin{equation}\label{eq.n5}
\E{\tilde{\gamma}_{[\rm {Shared.P}]}(L_j)\big\vert \Lambda_j=\lambda_{\rm high}}%> \lambda_{\rm high}
>\E{\tilde{\gamma}_{[\rm {Shared.P}]}(L_i)\big\vert \Lambda_i=\lambda_{\rm low}}, %< \lambda_{\rm low}.
\end{equation}
The equivalence between \eqref{eq.8} and \eqref{eq.n16} implies bias in the premium. Moreover,
Especially, \eqref{eq.n5} indicates that the {\it a priori} risk characteristics affect the {\it a posteriori} as well as the {\it a priori} risk classifications. The dual role of the {\it a priori} risk characteristics resulting in biased premiums is the double-counting problem.

In general, the unbiasedness in \eqref{eq.8} can be written as
\begin{equation}\label{e2}
  \E{\frac{M(\Lambda, L)}{\Lambda}\bigg\vert \Lambda=\lambda_{\kappa}}=1,
\end{equation}
where
\begin{equation*}
\frac{M(\Lambda, L)}{\Lambda}\quad\hbox{and}\quad \Lambda,
\end{equation*}
are called {\it pure relativity} and {\it pure priori rate}, respectively. Hence, \eqref{e2} is a condition for the unbiasedness of pure relativity.
Concerning the PPOS and PNO, we obtain the pure relativity as
\[
    \frac{\tilde{M}_{[{\rm Shared.P}]}(\lambda_\kappa, \ell)}{\Lambda}=\tilde{\gamma}_{[{\rm Shared.P}]}\left( \ell \right)\quad\hbox{and}\quad
        \frac{\tilde{M}_{[{\rm Shared.No}]}(\lambda_\kappa, \ell)}{\Lambda}=1.
\]
Now, from this pure relativity, we can formally define the double-counting problem as follows.

\begin{definition}\label{def.2}
The premium $M(\lambda_\kappa, \ell)$ suffers from the double-counting problem if pure relativity has the following order:
\[
\E{\frac{M(\Lambda, L)}{\Lambda}\bigg\vert \Lambda=\lambda_{\kappa_1}}>
\E{\frac{M(\Lambda, L)}{\Lambda}\bigg\vert \Lambda=\lambda_{\kappa_2}},
\quad\hbox{for any}\quad \lambda_{\kappa_1}>\lambda_{\kappa_2}.
\]
\end{definition}

We next provide a numerical example
of the double-counting problem.

\begin{example}\label{ex.1}

For the frequency random effects model, consider the specific distribution
\begin{equation*}%\label{eq.7}
        N_{t}\big\vert \left(\lambda_\kappa, \theta\right) \iid {\rm Pois}\left(\lambda_\kappa\,\theta\right),
\end{equation*}
with
\[
\P{\Lambda=\lambda_{\rm low}}=1/3,\quad  \P{\Lambda=\lambda_{\rm mid}}=1/3, \quad\hbox{and}\quad \P{\Lambda=\lambda_{\rm high}}=1/3,
\]
and
\[
\Theta\sim {\rm Gamma}(1, \psi),
\]
where $\Theta\sim {\rm Gamma}(1, \psi)$ represents a gamma distribution with mean of $1$ and dispersion parameter $\psi$.
Here, further to our real data analysis in Section \ref{sec.5}, we set $\psi=0.8$, $\lambda_{\rm low}=0.1$, $\lambda_{\rm mid}=0.5$, and $\lambda_{\rm high}=0.9$.

Now, we consider a PPOS under the BMS with the $-1/+2$ transition rule and $10$ BM levels.
The distribution of $L$ and the corresponding BM relativity are given in Table \ref{rl.ex3}.
We thus obtain the conditional expectations for pure relativity as follows:

\begin{equation*}%\label{eq.12}
  \begin{cases}
    \E{\frac{\tilde{M}_{[{\rm Shared.P}]}(\Lambda, L)}
    {\Lambda}\big\vert \Lambda=\lambda_{\rm low}  \,\,}=
    \E{\tilde{\gamma}_{[{\rm Shared.P}]}(L)\big\vert \Lambda=\lambda_{\rm low} \,\,}= 0.304;\\
    \E{\frac{\tilde{M}_{[{\rm Shared.P}]}(\Lambda, L)}
    {\Lambda}\big\vert \Lambda=\lambda_{\rm mid}\,}=
    \E{\tilde{\gamma}_{[{\rm Shared.P}]}(L)\big\vert \Lambda=\lambda_{\rm mid}\,}  = 0.805;\\
    \E{\frac{\tilde{M}_{[{\rm Shared.P}]}(\Lambda, L)}
    {\Lambda}\big\vert \Lambda=\lambda_{\rm high}}=
    \E{\tilde{\gamma}_{[{\rm Shared.P}]}(L)\big\vert \Lambda=\lambda_{\rm high}}= 1.069,\\
  \end{cases}
\end{equation*}
which in turn indicates the double-counting problem
\[
\E{\tilde{\gamma}_{[{\rm Shared.P}]}(L)\big\vert \Lambda=\lambda_{\rm low}} <
\E{\tilde{\gamma}_{[{\rm Shared.P}]}(L)\big\vert \Lambda=\lambda_{\rm mid}} <
\E{\tilde{\gamma}_{[{\rm Shared.P}]}(L)\big\vert \Lambda=\lambda_{\rm high}}.
\]
Furthermore, this bias in BM relativities leads to that in the premium:
\[
  \begin{cases}
    \E{\tilde{M}_{[{\rm Shared.P}]}(\Lambda, L)\big\vert \Lambda=\lambda_{\rm low}\,\,} =0.030; \\
    \E{\tilde{M}_{[{\rm Shared.P}]}(\Lambda, L)\big\vert \Lambda=\lambda_{\rm mid}\,}  =0.403; \\
    \E{\tilde{M}_{[{\rm Shared.P}]}(\Lambda, L)\big\vert \Lambda=\lambda_{\rm high}} 	  =0.962. \\
  \end{cases}
\]

\end{example}

\subsection{Measure for Double-Counting Problem}\label{sec.3.2}

Recall that an unbiased premium can be characterized by unbiasedness in pure relativity, as in \eqref{e2}. However, if the premium suffers from double-counting problem, the conditional expectation of pure premium
\begin{equation}\label{e3}
  \E{\frac{M(\Lambda, L)}{\Lambda}\bigg\vert \Lambda=\lambda_{\kappa}}
\end{equation}
would be an (increasing) function of $\lambda_{\kappa}$ by definition.
If the other conditions are identical, we would prefer the premium with smaller bias in \eqref{e3}, which is equivalent to having the smaller variance of pure relativity conditional expectation
\begin{equation}\label{eq.16}
\Var{ \E{\frac{M(\Lambda, L)}{\Lambda} \bigg\vert \Lambda}}.
\end{equation}
Because $\Var{\frac{M(\Lambda, L)}{\Lambda}}$ is decomposed into the summation of two components
\begin{equation*}%\label{eq.23}
\Var{\frac{M(\Lambda, L)}{\Lambda}}=
\Var{\E{\frac{M(\Lambda, L)}{\Lambda}\bigg\vert \Lambda}} + \E{\Var{\frac{M(\Lambda, L)}{\Lambda}\bigg\vert \Lambda}},\\
\end{equation*}
we propose the normalized version of \eqref{eq.16}
 \begin{equation}\label{eq.30}
{\rm FIX}\left(M\right):=\frac{\Var{\E{\frac{M(\Lambda, L)}{\Lambda}\bigg\vert \Lambda}}}{\Var{\frac{M(\Lambda, L)}{\Lambda}}}
 \end{equation}
 as the {\it fairness index} (FIX) of the premium $M(\lambda_\kappa, \ell)$.
 For the definition of \eqref{eq.30}, we consider only the FIX of the nontrivial premium $M(\lambda_\kappa, \ell)$ satisfying
 \[
\Var{\frac{M(\Lambda, L)}{\Lambda}}> 0.
 \]
  For premium $M(\lambda_\kappa, \ell)$, the FIX satisfies the inequality
 \[
 0\le {\rm FIX}(M) \le 1.
 \]
Note that no double-counting problem implies that ${\rm FIX}(M)=0$.
In contrast,
${\rm FIX}(M)=1$ implies
the presence of a perfect double-counting problem, meaning that
the {\it a posteriori} risk classification of the premium is determined purely by the {\it a priori} rate $\lambda_\kappa$.
Given that the other factors are equal, the BMS with a small value of ${\rm FIX}$ would be preferred.

As for the PPOS $\tilde{M}_{\rm [Shared.P]}\left(\lambda_\kappa, \ell \right)$, the FIX defined in \eqref{eq.16} can be simplified as follows, for a more intuitive interpretation:
 \begin{equation}\label{eq.17}
{\rm FIX}\left(\tilde{M}_{\rm [Shared.P]}\right)=\frac{\Var{\E{\tilde{\gamma}_{[{\rm Shared.P}]}(L)\big\vert \Lambda}}}{\Var{\tilde{\gamma}_{[{\rm Shared.P}]}(L)}}.
 \end{equation}
 The FIX in \eqref{eq.17} can be compared to the fairness measure defined in
  \citet{Chong}. For a specific comparison, see Remark \ref{rem.1} in the appendix.
 The FIX calculation procedure is given in Lemma \ref{lem.a3} in the appendix.

\begin{example}[Continued from Example \ref{ex.1}]\label{ex.2}

As regards the PPOS, we are mainly interested in the degree of double counting under various {\it a priori} rate distributions, and in how such double counting affects the quality of prediction under BMS.

In addition to the settings in Example \ref{ex.1},
we once again numerically study the various {\it a priori} rate distributions as summarized in Table \ref{tab.lamb}.
Compared to our base Scenario I, where the {\it a priori} rate distributions are the same as in Example \ref{ex.1}, Scenario II shows a smaller {\it a priori} rate variance with the same mean, whereas
Scenario III shows a larger {\it a priori} rate mean with the same variance.
 The Scenario IV {\it a priori} rates are left-skewed.
 All scenarios have distinct {\it coefficient of variation} (CV) values.

\begin{table}[h!]
\caption{The selected {\it a priori} rate sets}\label{tab.lamb}
\centering
\begin{tabular}{l c c c c c c c}
 \hline
 & $\lambda_{\rm low}$ & $\lambda_{\rm mid}$ & $\lambda_{\rm high}$ && mean & variance & CV\\ \hline
 Scenario I 		& 0.1 & 0.5 & 0.9 && 0.5 & 0.16 &0.80\\
 Scenario II 		& 0.4 & 0.5 & 0.6 && 0.5 & 0.01 &0.20\\
 Scenario III 	& 0.6 & 1.0 & 1.4 && 1.0 & 0.16 &0.40\\
 Scenario IV 		& 0.1 & 0.2 & 1.2 && 0.5 & 0.37 &1.22\\
 \hline
\end{tabular}
%}
\end{table}

The degree of double counting and quality of prediction are based on the FIX and HMSE
for each scenario, as in Table \ref{tab.ex.2}.
The distribution of $L$ is presented for a reference.
As in \citet{Taylor1997}, we find the set of the {\it a priori} rate with larger CV increasing the double-counting problem.

\begin{table}
\caption{FIX and HMSE under the -1/+2 system} \label{tab.ex.2}
\centering
\resizebox{\linewidth}{!}{%%% resize table
\begin{tabular}{ l  c c c c c c c c c c c||c c | c c}
 \hline
 	& \multicolumn{10}{c}{$\P{L=\ell}$} & &   \multicolumn{2}{c|}{$\tilde{M}_{\rm [Shared.P]}$} & \multicolumn{2}{c}{$\tilde{M}_{\rm [Shared.No]}$}\\ \cline{2-11}\cline{13-16}
	\multicolumn{1}{r}{$\ell$} & 1& 2& 3& 4& 5& 6& 7& 8& 9& 10 &&  FIX & HMSE &  FIX & HMSE\\
\hline
Scenario I 	&	0.414 	&	0.048 	&	0.059 	&	0.030 	&	0.032 	&	0.029 	&	0.036 	&	0.049 	
			&	0.087 	&	0.217 	&&	0.308 	& 0.163 & 0 & 0.285 \\
Scenario II	&	0.303 	&	0.050 	&	0.064 	&	0.039 	&	0.043 	&	0.041 	&	0.051 	&	0.068 	
			&	0.112 	&	0.230 	&&	0.018 	& 0.099 & 0 & 0.205\\
Scenario III	&	0.170 	&	0.031 	&	0.040 	&	0.026 	&	0.030 	&	0.032 	&	0.043 	&	0.066 	
			&	0.133 	&	0.427 	&&	0.063 	& 0.560 & 0 & 0.885\\
Scenario IV	&	0.492 	&	0.053 	&	0.063 	&	0.029 	&	0.029 	&	0.024 	&	0.027 	&	0.035 	
			&	0.062 	&	0.187 	&&	0.469 	& 0.256 & 0 & 0.397\\
\hline
\end{tabular}
}
\end{table}

\end{example}

Unlike the PPOS, the PNO under any BMS is unbiased and free of the double-counting problem since the equation shows that
\[
\begin{aligned}
\E{\tilde{M}_{\rm [Shared.No]}\left(\Lambda_i, \mathcal{F}_{i,t}\right)\big\vert \Lambda_i=\lambda_\kappa}
&=\lambda_\kappa;
\end{aligned}
\]
this can be confirmed with the FIX of zero.
We know that the introduction of the {\it a posteriori} risk classification can improve the PNO's prediction ability; for an example, see the PPOS in \eqref{eq.p1}. However,
the introduction of the {\it a posteriori} risk classification can lead to the double-counting problem in PPOS. In summary, the PPOS outperforms the PNO in terms of prediction power at the cost of double counting. The last two columns in Table \ref{tab.ex.2} provide the FIX and HMSE for the PNO under the same settings in Example \ref{ex.2}. A comparison of the FIX and HMSE in the PPOS and PNO in Table \ref{tab.ex.2} confirms our findings.

\begin{remark}
  We defines the FIX in \eqref{eq.30} for premium $M$ in the BMS, but the FIX does not have to be defined under the BMS. For example, for the Bayesian premium $M_{\rm [Bayes]}$ in \eqref{eq.6}, the FIX can be defined as
  \[
  {\rm FIX}(M_{\rm [Bayes]}):=\frac{\Var{\E{M_{\rm [Bayes]}\big\vert \Lambda}}}{\Var{M_{\rm [Bayes]}}}.
  \]
\end{remark}

In Section \ref{sec.3}, we argue that the Bayesian premium $M_{\rm [Bayes]}(\boldsymbol{X}, \mathcal{F}_t)$ in \eqref{eq.6} and the Bayesian estimator of the residual effect $\E{\Theta\big\vert \boldsymbol{X}, \mathcal{F}_t}$ are unbiased owing to \eqref{eq.41} and \eqref{eq.42}. This implies that the experience ratemaking based on Bayesian estimation is free of the double-counting problem.
The following corollary confirms this with the FIX.

\begin{corollary}\label{cor.2}
  The Bayesian premium $M_{\rm [Bayes]}(\Lambda, \mathcal{F}_t)$ satisfies
  \[
  {\rm FIX}(M_{\rm [Bayes]})=0%\quad\hbox{and}\quad {\rm FIX}(\tilde{M}_{\rm [Shared.P]})=0.
  \]

\end{corollary}
\begin{proof}
We prove this corollary as follows:
  \[
  \begin{aligned}
  \E{\frac{M_{\rm [Bayes]}(\Lambda, \mathcal{F}_t)}{\Lambda}\bigg\vert \Lambda}
  &=\E{\E{\Theta\big\vert \Lambda, \mathcal{F}_t}\big\vert \Lambda}\\
  &=\E{\Theta\big\vert \Lambda}\\
  &=1,
  \end{aligned}
  \]
  where the last equality is from the independence assumption and \eqref{eq.3}.
\end{proof}

\begin{remark}\label{rem.3} So far, we considered the premiums in the classical BMS literature as having nonzero FIX, and the Bayesian premium as having zero FIX. While both the premiums have the same {\it a priori} rate, the former suffers from double counting, but the latter does not, because of the {\it a priori} risk characteristics in their {\it a posteriori} risk ratemaking function. To be more specific, the premiums in the classical BMS literature are directly exposed to double counting because they use the same BM relativity table for all the policyholders; that is, the policyholders' BM levels are determined by their claim histories, which are inevitably influenced by their {\it a priori} risk characteristics. Therefore, a zero FIX is attainable only when the {\it a posteriori} risk classification does not depend solely on the policyholders' claim history. That is, the {\it a posteriori} risk classification should also use the policyholders' individual risk characteristics to remove its own influence due to their claim history, which is the key to zero FIX of the Bayesian premium. In other words, the individualized {\it a posteriori} risk classification is critical to prevent the double-counting problem. In this regard, \citet{Pitrebois2003} proposed the use of several BM relativity tables to reduce the double-counting problem.
\end{remark}
In the following subsection, we investigate the double-counting problem of the BMS with several BM relativity tables based on the {\it a priori} risk characteristics.

\subsection{Double-Counting Problem and the BMS with Individualized BM Relativity Table} \label{sec.3.3}

Using numerical examples, \citet{Pitrebois2003} show that the introduction of several BM relativity tables
based on a few {\it a priori} risk characteristics can reduce the double-counting problem.
In this subsection, we formally present the BMS in which all policyholders have their own BM relativity tables according to their {\it a priori} risk characteristics, and show that such a BMS does not suffer from the double-counting problem at all. First, we start with the specification of such a BMS.

 \begin{BMS}[BMS with Individualized BM Relativity Table]\label{BMS.2}
  Consider the BMS with the following properties:

\begin{itemize}
  \item {\bf A Priori Rate}: A policyholder with the {\it a priori} risk characteristics $\boldsymbol{x}_{\kappa}$ given the {\it a priori} rate
  \[
\xi_{\rm[Ind]}(\lambda_{\kappa})
  \]
  at the time of contract.
  \item {\bf Relativity Table}: All policyholders have their own individualized BM relativity tables depending on their {\it a priori} risk characteristics
  \begin{equation}\label{eq.rel.1}
\gamma_{\rm [Ind]}(\lambda_{\kappa},1),\gamma_{\rm [Ind]}(\lambda_{\kappa},2), \cdots, \gamma_{\rm [Ind]}( \lambda_{\kappa},z).%,\quad {\kappa}=1%, \cdots, \mathcal{K}.
\end{equation}

 \item {\bf Premium}: The premium of a policyholder at the BM level $\ell$ with the {\it a priori} risk characteristics $\boldsymbol{x}_{\kappa}$ obtained by multiplying the {\it a priori} rate with the relativity
\begin{equation*}%\label{eq.17}
\begin{aligned}
M_{\rm [Ind]}(\lambda_{\kappa}, \ell)&=\xi_{\rm [Ind]}(\lambda_{\kappa})\gamma_{\rm [Ind]}(\lambda_{\kappa}, \ell).\\
\end{aligned}
\end{equation*}
\end{itemize}

\end{BMS}

Now, under the {\it BMS with individualized BM relativity table}, we explain how to obtain the individualized BM relativities in \eqref{eq.rel.1} for all the risk characteristics $\boldsymbol{x}_\kappa$. The aim is again to minimize the mean square error in Criteria \ref{criteria.1}.
As in the previous section, for the sake of brevity, we choose the {\it a priori} rate function
\[
\tilde{\xi}_{\rm[Ind]}(\boldsymbol{x}_{\kappa}):=\eta^{-1}(\boldsymbol{x}_{\kappa}\boldsymbol{\beta})
\quad\hbox{or equivalently}\quad
\tilde{\xi}_{\rm[Ind]}(\lambda_{\kappa}):=%\eta^{-1}(\boldsymbol{x}_{\kappa}\boldsymbol{\beta})=
\lambda_{\kappa}.
\]
Then, for each risk group $\kappa=1, \cdots, \mathcal{K}$, we choose the BM relativity in \eqref{eq.rel.1} as the solution to the optimization problem
\begin{equation}\label{eq.ahn13}
\tilde{\gamma}_{\rm [Ind]}(\lambda_{\kappa}, \ell) :=\argmin_{\gamma_{\rm [Ind]}} \mathbb{E}[\left(\Lambda\Theta-\Lambda\gamma_{\rm [Ind]}(\Lambda, L)\right)^2],
\end{equation}
where the RHS of \eqref{eq.ahn13} is optimized with the non-negative functions.
For convenience, we also define
\begin{equation}\label{eqq.1}
\tilde{M}_{\rm [Ind]}(\lambda_{\kappa}, \ell):=\tilde{\xi}_{\rm [Ind]}(\lambda_{\kappa}) \tilde{\gamma}_{\rm [Ind]}(\lambda_{\kappa}, \ell).
\end{equation}

\bigskip
\noindent
{{\bf Premium from Optimizing the Individualized BM Relativity Table (POI).}
From among the various types of premiums in the {\it BMS with individualized BM relativity table}, consider the policyholder's premium at the BM level $\ell$ with the {\it a priori} rate $\lambda_\kappa$ given in \eqref{eqq.1}. This premium will be called the {\it premium from optimization of individualized BM relativity table} (POI).
\bigskip
}

\begin{theorem}\label{eff.thm.1}
Equation \eqref{eq.ahn13} is solved as
\begin{equation}\label{eff.eq.7}
       \tilde{\gamma}_{\rm [Ind]}\left(\lambda_{\kappa}, \ell\right)=\E{\Theta\big\vert \Lambda=\lambda_{\kappa}, {L}=\ell},
\end{equation}
  with
\begin{equation*}%\label{eff.eq.6}
    \begin{aligned}
      \E{\Theta\big\vert \Lambda=\lambda_{\kappa}, L=\ell}=\frac{\int_0^{\infty} \theta\pi_{\ell}(\lambda_{\kappa} \theta,  \psi ) g(\theta) {\rm d} \theta } {\int_0^{\infty} \pi_{\ell}(\lambda_{\kappa} \theta,  \psi ) g(\theta) {\rm d} \theta }.
    \end{aligned}
\end{equation*}
\end{theorem}
\begin{proof}
First, we obtain \eqref{eff.eq.7} from the derivative of
\[
\mathbb{E}[\left(\Lambda\Theta-\Lambda\gamma_{\rm [Ind]}(\Lambda, L)\right)^2]
=\sum\limits_{\ell}\sum\limits_{\kappa}
\mathbb{E}[\left(\Lambda\Theta-\Lambda\gamma_{\rm [Ind]}(\Lambda, L)\right)^2\big\vert
L=\ell, \boldsymbol{X}=\boldsymbol{x}_\kappa]\P{\Lambda=\lambda_\kappa, L=\ell}
\]
with respect to ${\gamma}_{\rm [Ind]}\left(\lambda_{\kappa}, \ell\right)$.
Furthermore, the RHS of \eqref{eff.eq.7} can be calculated from the expression
\[
\begin{aligned}
\E{\Theta\big\vert \Lambda=\lambda_{\kappa}, L=\ell}&=\int_0^{\infty} \theta \P{{L}=\ell\big\vert \Lambda=\lambda_{\kappa}, \Theta=\theta} g(\theta) {\rm d}\theta\frac{\P{\Lambda=\lambda_{\kappa}}}{\P{ \Lambda=\lambda_{\kappa}, L=\ell}}\\
&=\frac{\int_0^{\infty} \theta  \P{{L}=\ell\big\vert \Lambda=\lambda_{\kappa}, \Theta=\theta} g(\theta) {\rm d}\theta}{\P{{L}=\ell\big\vert \Lambda=\lambda_{\kappa}}}\\
&=\frac{\int_0^{\infty} \theta \pi_j(\lambda_{\kappa}\theta,  \psi ) g(\theta) {\rm d}\theta}{\int_0^{\infty}  \pi_j(\lambda_{\kappa}\theta,  \psi ) g(\theta) {\rm d}\theta},
\end{aligned}
\]
where the first equation is just a rearrangement of the conditional distribution, and the last equation is obtained from the definition of the stationary distribution of ${L}$.

\end{proof}

Recall that the double-counting problem arises from the shared BM relativity table, which does not consider the {\it a priori} risk characteristics as a
double-counting problem.
In fact, the observation
 \[
 \begin{aligned}
   \E{\tilde{M}_{\rm [Ind]}(\Lambda, L)\vert \Lambda}&=\Lambda\E{\tilde{\gamma}_{\rm [Ind]}(\Lambda, L)\big\vert \Lambda}\\
   &=\Lambda
 \end{aligned}
 \]
 indicates removal of the bias.
 Here, the second inequality is from the unbiasedness of relativity
  \[
 \begin{aligned}
   \E{\tilde{\gamma}_{\rm [Ind]}(\Lambda, L)\vert \Lambda}&=\E{\E{\Theta\big\vert \Lambda, L}\big\vert \Lambda}\\
   &=\E{\Theta\big\vert \Lambda}\\
   &=\E{\Theta}\\
   &=1,
 \end{aligned}
 \]
 where the third equality is from the assumption of independence between $\Lambda$ and $\Theta$.
 Furthermore, the following calculation of FIX confirms our claim:
 \[
 \begin{aligned}
 {\rm FIX}[\tilde{M}_{\rm [Ind]}]&=\frac{\Var{\E{\tilde{\gamma}_{\rm [Ind]}\left(\Lambda, {L}\right)  \big\vert \Lambda}}}{\Var{ \tilde{\gamma}_{\rm [Ind]}\left(\Lambda, {L}\right) }}\\
&=\frac{\Var{\E{\E{\Theta\big\vert \Lambda, {L}}}\big\vert \Lambda}}{\Var{   \E{\Theta\big\vert \Lambda, {L}} }}\\
&=\frac{\Var{\E{\Theta\big\vert \Lambda}}}{\Var{   \E{\Theta\big\vert \Lambda, {L}} }}\\
&=0,
\end{aligned}
 \]
 where the last equality also comes from the independence between $\Lambda$ and $\Theta$.

In this subsection, we show that the double counting in the BMS can be removed with the individualized BM relativity table. However, the individual system is too complex to be presented to consumers.
Therefore, we discuss how to address the double-counting problem
of the {\it BMS using the shared BM relativity table} in the traditional BMS format.

\section{Full Optimization and Double Counting}\label{sec.4}
In terms of Criteria \ref{criteria.1},
the premiums of the {\it BMS with shared BM relativity table} allow for utilizing the {\it a priori} rate function ${\xi}_{[{\rm Shared}]}$ as well as the BM relativity function ${\gamma}_{[{\rm Shared}]}$ in the optimization process.
For convenience, we predetermine the {\it a priori} rate function in PPOS as
$$\xi_{[{\rm Shared}]}(\lambda_{\kappa})=\eta^{-1}\left( \boldsymbol{x}_{k}\boldsymbol{\beta} \right)$$
before the optimization of $\gamma_{[{\rm Shared}]}$ in \eqref{eq.ahn130}.
However, in terms of optimization, the {\it a priori} rate function $\xi_{[{\rm Shared}]}$ need not be predetermined for the optimization process, and can achieve a better optimization result by allowing for the {\it a priori} rate function to participate in the optimization process.
 In this section, we show how to utilize the {\it a priori} rate function $\xi_{[{\rm Shared}]}$ as well as the BM relativity function $\gamma_{[{\rm Shared}]}$ for optimization with the {\it BMS with shared BM relativity table}, and study how such inclusion affects the double-counting problem.

\subsection{Premium from Full Optimization of Shared BM Relativity Table}\label{sec.4.1}
 This subsection considers the premium of the {\it BMS with shared BM relativity table}, allowing for the {\it a priori} rate function to participate in the optimization process.
 From Criteria \ref{criteria.1}, the objective function in \eqref{eq.important} can be
further minimized by allowing for the {\it a priori} rate function $\xi_{[{\rm Shared}]}$ to contribute to the optimization process as well as the BM relativity function ${\gamma}_{[{\rm Shared}]}$. Specifically, we try to solve the optimization problem

\begin{equation}\label{eq.a.7}
\min_{ {\xi}, {\gamma} } \E{\left(\Lambda\Theta- {\xi}(\Lambda){\gamma}({L}) \right)^2},
\end{equation}
where $\xi$ and $\gamma$ are optimized from among the non-negative functions.
Clearly, the solution of \eqref{eq.a.7} is not unique. For example, if $(\tilde{\xi}, \tilde{\gamma})$ is the solution to \eqref{eq.a.7}, $\left(\frac{1}{c}\tilde{\xi}, c\,\tilde{\gamma} \right)$ will also be the solution to \eqref{eq.a.7} for any $c>0$.
To address this trivial difficulty, we provide the constrained version of \eqref{eq.a.7} as follows:

\begin{equation}\label{eq.a.70}
\min_{ {\xi_q}, {\gamma_q} } \E{\left(\Lambda\Theta- {\xi_q}(\Lambda){\gamma_q}({L}) \right)^2},
 \quad\hbox{subject to}\quad {\gamma_q} \left( \left\lfloor \frac{z}{2} \right\rfloor \right)=q
\end{equation}
for some predetermined constant $q>0$, where $\left\lfloor\cdot\right\rfloor$ is a floor function.
For the pair $(\tilde{\xi}_{[{\rm Shared.F}]}, \tilde{\gamma}_{[{\rm Shared.F}]})$ of the solution of \eqref{eq.a.70},
we define the premium as
\begin{equation}\label{eq.71}
\tilde{M}_{[{\rm Shared.F}]}(\lambda_\kappa, \ell):= \tilde{\xi}_{[{\rm Shared.F}]}(\lambda_\kappa)  \tilde{\gamma}_{[{\rm Shared.F}]}(\ell).
\end{equation}
 Note that the constraint in \eqref{eq.a.70} is actually not a constraint in terms of optimization efficiency, and one can easily show that the two optimization problems of \eqref{eq.a.7} and \eqref{eq.a.70} result in the same efficiency,
\begin{equation}\label{eq.a3}
\min_{ {\xi}, {\gamma} } \E{\left(\Lambda\Theta- {\xi}(\Lambda){\gamma}({L}) \right)^2}
=\min_{ {\xi_q}, {\gamma_q} } \E{\left(\Lambda\Theta- {\xi_q}(\Lambda){\gamma_q}({L}) \right)^2},
\end{equation}
for any $q>0$ with ${\gamma_q} \left( \left\lfloor \frac{z}{2} \right\rfloor \right)=q$.
Next, we formally define the premium obtained from full optimization of the shared BM relativity table.

\bigskip
\noindent
{{\bf Premium obtained from full optimization of shared BM relativity table (PFOS).}
From among the various types of premiums of the {\it BMS with shared BM relativity table}, we consider the policyholder's premium at the BM level $\ell$ with the {\it a priori} rate $\lambda_\kappa$ as given in \eqref{eq.71}.
This is the {\it PFOS}, meaning that both the BM relativity and {\it a priori} rate function are determined by optimizing the {\it BMS with shared relativity table}.
\bigskip
}

Clearly, the PFOS is still a premium obtained from the {\it BMS with shared BM relativity table} since its risk classification includes the following:%is consist of
\begin{itemize}
  \item The a priori classification given by $\xi_{[{\rm Shared}]}(\lambda_\kappa) = \tilde{\xi}_{\rm [Shared.F]}(\lambda_\kappa)$, determined by the {\it a priori} risk characteristics.
  \item The {\it a posteriori} classification function given by $\gamma_{[{\rm Shared}]}(\ell) = \tilde{\gamma}_{\rm [Shared.F]}(\ell)$, determined at the BM level $\ell$.
\end{itemize}
While \eqref{eq.a.70} cannot be analytically solved in general, it can be calculated with Algorithm \ref{algo.a.1} given in the appendix, where the coordinate descent algorithm is presented. A detailed discussion of the algorithm is presented in the appendix.

In fact, the PFOS can also be interpreted as the premium in the
BMS with individualized BM relativities:
\begin{equation}\label{eq.bms10}
\begin{aligned}
\tilde{M}_{\rm [Shared.F]}(\ell, \lambda_{\kappa})&=\tilde{\xi}_{\rm [Shared.F]}(\lambda_{\kappa})\tilde{\gamma}_{\rm [Shared.F]}(\ell)\\
&= \lambda_{\kappa}\left( \frac{\tilde{\xi}_{\rm [Shared.F]}(\lambda_{\kappa})\tilde{\gamma}_{\rm [Shared.F]}(\ell)}{\lambda_{\kappa}}\right),\\
\end{aligned}
\end{equation}
where the {\it a priori} rate function is given as
$\xi_{\rm [Ind]}(\lambda_{\kappa})=\lambda_\kappa$
and the BM relativity table for the policyholder with the {\it a priori} rate $\lambda_{\kappa}$ is given by
\begin{equation*}%\label{eq.a1}
\left({\gamma}_{\rm [Ind]}(1, \lambda_\kappa), \cdots, {\gamma}_{\rm [Ind]}(z, \lambda_\kappa) \right) = \left(\frac{\tilde{\xi}_{\rm [Shared.F]}(\lambda_{\kappa})}{\lambda_{\kappa}}\tilde{\gamma}_{\rm [Shared.F]}(1), \cdots, \frac{\tilde{\xi}_{\rm [Shared.F]}(\lambda_{\kappa})}{\lambda_{\kappa}}\tilde{\gamma}_{\rm [Shared.F]}(z) \right).
\end{equation*}
Note that this individualized BM relativity can be interpreted as pure relativity,
for which we can conveniently provide the following notation:
\begin{equation}\label{eq.b5}
\tilde{\gamma}_{\rm [Shared.F]}^*(\lambda_\kappa, \ell):=\frac{\tilde{\xi}_{\rm [Shared.F]}(\lambda_{\kappa})}{\lambda_{\kappa}}\tilde{\gamma}_{\rm [Shared.F]}(\ell)
\quad\hbox{and}\quad
\tilde{\xi}_{\rm [Shared.F]}^*(\lambda_\kappa):=\lambda_\kappa.
\end{equation}
The optimization of \eqref{eq.a.70} resembles the optimization of the POI since it can be
equivalently written as
\begin{equation}\label{eq.aa}
\min_{ {\xi}_{[{\rm Shared}]}, {\gamma}_{[{\rm Shared}]} } \E{\left(\Lambda\Theta- {\xi}_{[{\rm Shared}]}(\Lambda){\gamma}_{[{\rm Shared}]}({L}) \right)^2}
=\min_{\gamma^* } \E{\left(\Lambda\Theta- \Lambda \gamma_{\rm [Shared.F]}^*(\Lambda, L)\right)^2},
\end{equation}
where the RHS of the equality is optimized among the functions with form $\gamma_{\rm [Shared.F]}^*$, which can be written as
\begin{equation}\label{eq.aa.1}
\gamma_{\rm [Shared.F]}^*(\lambda_\kappa, \ell)=f_1(\lambda_\kappa)f_2(\ell)
\end{equation}
for some positive functions $f_1$ and $f_2$ satisfying

\begin{equation}\label{eq.res.1}
f_2\left(\left\lfloor \frac{z}{2}\right\rfloor\right) = q,
\end{equation}
for some predetermined constant $q>0$.
Thus, the optimization of \eqref{eq.aa} is similar to that of \eqref{eq.a.70}, where the BM relativity function has a further restriction: it is optimized among the functions in \eqref{eq.aa.1}.\footnote{Note that \eqref{eq.res.1} is really not a restriction, but is introduced to avoid trivial non-unique solutions.}
In this context, we can obtain the following result:
\begin{lemma}\label{lem.1}
  The optimized value has the following relation.
  \[
    \E{\left(\Lambda\Theta- \tilde{M}_{\rm [Ind]}(\Lambda,L) \right)^2}
 \le
    \E{\left(\Lambda\Theta- \tilde{M}_{\rm [Shared.F]}(\Lambda,L) \right)^2}
   \le
  \E{\left(\Lambda\Theta- \tilde{M}_{[{\rm Shared.P}]}(\Lambda,L) \right)^2}
  \]
\end{lemma}

\begin{proof}
  We can immediately obtain the proof from the optimization settings in \eqref{eq.ahn130}, \eqref{eq.ahn13}, and \eqref{eq.a.70}.
\end{proof}

In summary, we can view the PFOS as a premium with individualized BM relativity table having restrictions on the BM relativity function, as in \eqref{eq.aa.1}. Because the POI can perfectly solve the double-counting problem, we expect the PFOS to partially solve the double-counting problem. This solution will depend on how the BM relativity function in the POI is similar to that of the PFOS, or how the multiplicative assumption in \eqref{eq.aa.1} deviates
from the BM relativity function in the POI.
The following numerical experiments compare the performance of the PFOS with that of the PPOS and POI in terms of the double counting as well as prediction power.

\begin{example}[Continued from Example \ref{ex.1}]\label{ex.3}
Under the same setting of Example \ref{ex.1}, we consider the various premiums for the PNO, PPOS, PFOS, and PIO.
In order to calculate the {\it a priori} rate function and BM relativities in the PFOS, we set the restriction as
\[
q=\tilde{\gamma}_{[\rm Shared.P]}  \left( \left\lfloor \frac{z}{2} \right\rfloor \right)
\]
in \eqref{eq.a.70}
for a unique solution to the full optimization problem.
Table \ref{rl.ex3} presents the distributions of $L$ and the BM relativities,
and Table \ref{pri.ex3} presents the {\it a priori} rate function for each scenario.
To compare the fairness and performance of the premium,
we provide the conditional expectation of relativity given the {\it a priori} rate, FIX, and HMSE as shown in Table \ref{dc.ex3}.

%We are mainly interested in comparing the PPOS, PFOS, and POI.
We first observe that
the PPOS shows the worst performance in terms of both the FIX and HMSE. Obviously,
the PPOS suffers from the double-penalty problem, as in Example \ref{ex.1}
and the conditional mean of the BM relativity in Table \ref{dc.ex3}.
It seems that the poor HMSE of PPOS is due to the double-penalty problem.
Moreover, Table \ref{dc.ex3} also shows that the PFOS has improved over the PPOS in terms of both HMSE and FIX.
The improvement in HMSE can be expected from Lemma \ref{lem.1},
although we find a rather unexpected dramatic improvement in FIX.
For further comparison of the PPOS and PFOS,
we match the relativity and the {\it a priori} rate pairs of two BMSs,
$\left(\tilde{\gamma}_{[\rm Shared.P]}, \tilde{\xi}_{[\rm Shared.P]} \right) $ and %\quad\hbox{and}\quad
$\left(\tilde{\gamma}_{[\rm Shared.F]}, \tilde{\xi}_{[\rm Shared.F]} \right)$.
While the two relativities
$\tilde{\gamma}_{[\rm Shared.P]}(\ell)$ and $\tilde{\gamma}_{[\rm Shared.F]}(\ell)$ show no much difference,
the major differences between the two BMSs are in the {\it a priori} rates. Specifically, the PFOS is biased in that
the conditional expectation of the relativities deviate from 1:
\begin{equation}\label{eq.c2}\begin{cases}
\E{\tilde{\gamma}_{[{\rm Shared.F}]}(L) | \Lambda=\lambda_{\rm low}  \,\,}= 0.291;\\
\E{\tilde{\gamma}_{[{\rm Shared.F}]}(L) | \Lambda=\lambda_{\rm mid}\,}= 0.814;  \\
\E{\tilde{\gamma}_{[{\rm Shared.F}]}(L) | \Lambda=\lambda_{\rm high} }= 1.089, \\
\end{cases}
\end{equation}
which is very similar to the case of PPOS,
\[
\begin{cases}
\E{\tilde{\gamma}_{[{\rm Shared.P}]}(L) | \Lambda=\lambda_{\rm low}  \,\,}= 0.304;\\
\E{\tilde{\gamma}_{[{\rm Shared.P}]}(L) | \Lambda=\lambda_{\rm mid}\,}= 0.805;  \\
\E{\tilde{\gamma}_{[{\rm Shared.P}]}(L) | \Lambda=\lambda_{\rm high} }= 1.069. \\
\end{cases}
\]
However, the bias in the {\it a priori} rates of PFOS compensates for the bias in \eqref{eq.c2},
\[
  \tilde{\xi}_{[{\rm Shared.F}]}(\lambda_{\rm low}) = 0.32 > 0.1\quad\hbox{and}\quad
  \tilde{\xi}_{[{\rm Shared.F}]}(\lambda_{\rm high}) = 0.84 < 0.9,
\]
and almost resolves the double-penalty problem, as the following conditional expectation of each premium shows:
\begin{equation*}
\begin{cases}
\E{\tilde{M}_{[{\rm Shared.F}]}(\Lambda, L) | \Lambda=\lambda_{\rm low}  \,\,}
=\lambda_{\rm low} \,\,\E{\tilde{\gamma}_{[{\rm Shared.F}]}(L) | \Lambda=\lambda_{\rm low}  \,} = 0.094; \\
\E{\tilde{M}_{[{\rm Shared.F}]}(\Lambda, L) | \Lambda=\lambda_{\rm mid} \,}
=\lambda_{\rm mid} \,\E{\tilde{\gamma}_{[{\rm Shared.F}]}(L) | \Lambda=\lambda_{\rm mid}\,} = 0.484;\\
\E{\tilde{M}_{[{\rm Shared.F}]}(\Lambda, L) | \Lambda=\lambda_{\rm high} }
=\lambda_{\rm high} \E{\tilde{\gamma}_{[{\rm Shared.F}]}(L) | \Lambda=\lambda_{\rm high} } =  0.913. \\
\end{cases}
\end{equation*}

Next, we compare the performance of the PFOS and POI. For a fair comparison, we match the pure relativity and the {\it a priori} rate function pairs of two BMSs:
$\left(\tilde{\gamma}_{[\rm Shared.F]}^*, \tilde{\xi}_{[\rm Shared.F]}^* \right)$ and % \quad\hbox{and}\quad
$\left(\tilde{\gamma}_{[\rm Ind]}, \tilde{\xi}_{[\rm Ind]} \right)$.
While the {\it a priori} rate is the same for both the BMSs in terms of representation, pure relativity shows a considerable difference between the two BMSs; this is the only factor affecting the difference in HMSE between the two BMSs. Despite the clear difference in relativity, the two BMSs do not show much difference in FIX, which can be explained by the hardly biased conditional pure relativity means of the PFOS:
\[
\begin{cases}
\E{\tilde{\gamma}_{[{\rm Shared.F}]}^*(\Lambda, L) | \Lambda=\lambda_{\rm low}  \,\,}= 0.943;\\
\E{\tilde{\gamma}_{[{\rm Shared.F}]}^*(\Lambda, L) | \Lambda=\lambda_{\rm mid} \,}= 0.969;  \\
\E{\tilde{\gamma}_{[{\rm Shared.F}]}^*(\Lambda, L) | \Lambda=\lambda_{\rm high} }= 1.015. \\
\end{cases}
\]

Finally, we interpret the $\triangle$ FIX and $\triangle$ HMSE columns in Table \ref{dc.ex3}.
The $\triangle$ FIX column shows the sequential changes in $\%$ of four BMSs, normalized by maximum FIX, in the order
\[
\hbox{PNO}\quad\rightarrow\quad\hbox{PPOS}\quad\rightarrow\quad\hbox{PFOS}\quad\rightarrow\quad\hbox{POI};
\]
the $\triangle$ HMSE is similarly defined.
We examine the difference from the PNO to PPOS.
While the HMSE of PPOS shows a significant improvement over that of PNO,
the PPOS creates the double-penalty problem with a FIX of 0.3075.
We also find significant improvement in the double-counting problem from the PPOS to PFOS, where the $\triangle$ FIX is $-99.3\%$, and minor improvement in HMSE, where the $\triangle$ HMSE is $-2.3\%$, and minor improvement in both the FIX and HMSE from the PFOS to PPOS.

\begin{table}[h]
\caption{Various BM relativities and distribution of $L$ under the -1/+2 system in Example \ref{ex.3}}\label{rl.ex3}
\centering
\resizebox{\linewidth}{!}{%%% resize table
%\resizebox{!}{\textheight}{%%% resize table
\begin{tabular}{ l l*{25}{c}}
 \hline
 	&Risk       			& \multicolumn{10}{c}{Level $\ell$}   \\ \cline{3-12}
 	&group $\kappa$	& 1& 2& 3& 4& 5& 6& 7& 8& 9& 10\\% &&  FIX & HMSE\\
\hline
\rowcolor{Gray}
$\tilde{\gamma}_{[{\rm Shared.No}]}$&
&1 	&	1 	&	1 	&	1 	&	1 	&	1 	&	1 	&	1	&	1 	&	1 \\%&&	0 	&	0.2853 	\\
$\tilde{\gamma}_{[{\rm Shared.P}]}$&
&0.240 	&	0.364 	&	0.390 	&	0.489 	&	0.544 	&	0.647 	&	0.752 	&	0.913 	&	1.160 	&	1.675 \\%	&&	0.3075 	&	0.1629 	\\
\rowcolor{Gray}
$\tilde{\gamma}_{\rm [Shared.F]}$&&	
0.224 	&	0.357 	&	0.382 	&	0.488 	&	0.544 	&	0.651 	&	0.759 	&	0.926 	&	1.183 	&	1.722\\% 	&&	0.0022 	&	0.1563 	\\
$\tilde{\gamma}_{\rm [Shared.F]}^*$ & low&	
0.724 	&	1.156 	&	1.237 	&	1.579 	&	1.761 	&	2.108 	&	2.458 	&	2.998 	&	3.830 	&	5.576\\% 	&&	0.0022 	&	0.1563 	\\
&mid	&	
0.266 	&	0.425 	&	0.455 	&	0.580 	&	0.647 	&	0.774 	&	0.903 	&	1.102 	&	1.407 	&	2.049\\% && & \\
&high	&	
0.208 	&	0.333 	&	0.356 	&	0.454 	&	0.507 	&	0.606 	&	0.707 	&	0.863 	&	1.102 	&	1.604\\% && &\\
\rowcolor{Gray}
$\tilde{\gamma}_{\rm [Ind]}$
& low	&	0.763 	&	1.330 	&	1.397 	&	1.888 	&	2.041 	&	2.457 	&	2.694 	&	3.059 	&	3.361 	&	3.725\\% 	&&	0 	&	0.1553 	\\
\rowcolor{Gray}
&mid&	
0.296 	&	0.475 	&	0.511 	&	0.658 	&	0.737 	&	0.884 	&	1.024 	&	1.228 	&	1.505 	&	1.957\\% 	&&		&		\\
\rowcolor{Gray}
&high&	
0.180 	&	0.286 	&	0.309 	&	0.398 	&	0.451 	&	0.547 	&	0.651 	&	0.813 	&	1.072 	&	1.626\\% 	&&		&		\\
\hline\hline	
$\P{L=\ell}$ && 0.414& 0.048& 0.059& 0.030& 0.032& 0.029& 0.036& 0.049& 0.087& 0.217 \\		
\hline												
\end{tabular}
} %%%
\end{table}

\begin{table}[h!]
\caption{The {\it a priori} rate functions in Example \ref{ex.3}}\label{pri.ex3}
\centering
\begin{tabular}{ l l*{25}{c}}
 \hline
  	& \multicolumn{3}{c}{Risk group $\kappa$}  \\ \cline{2-4}
	& low& mid& high \\
\hline
\rowcolor{Gray}
$\tilde{\xi}_{[{\rm Shared.No}]}$ 	& 0.10 	&	0.50 	&	0.90 	\\	
$\tilde{\xi}_{[{\rm Shared.P}]}$ 	& 0.10 	&	0.50 	&	0.90 	\\						
\rowcolor{Gray}
$\tilde{\xi}_{[{\rm Shared.F}]}$	& 0.32 	&	0.59 	&	0.84 	\\		
$\tilde{\xi}_{[{\rm Shared.F}]}^*$	& 0.10	&	0.50 	&	0.90 	\\		
\rowcolor{Gray}
$\tilde{\xi}_{[{\rm Ind}]}$	 		& 0.10 	&	0.50 	&	0.90 	\\									
\hline												
\end{tabular}
\end{table}

\begin{table}[h!]
\caption{Summary statistics of each BMS in Example \ref{ex.3}}\label{dc.ex3}
\centering
\resizebox{\linewidth}{!}{%%% resize table
\begin{tabular}{ l l*{25}{c}}
 \hline
 Method    && $\E{\tilde{\gamma}| \Lambda=\lambda_{\rm low} }$&
 $\E{\tilde{\gamma}| \Lambda=\lambda_{\rm mid}}$&
 $\E{\tilde{\gamma}| \Lambda=\lambda_{\rm high} }$ &&
 FIX & $\triangle$ FIX & HMSE & $\triangle$ HMSE \\
\hline
\rowcolor{Gray}
Shared.No 	 	&& 1		& 1		& 1		&& 0 	 & - 	    & 0.2853 & -\\	
Shared.P 		 	&& 0.304 	& 0.805	& 1.069	&& 0.3075 &  100$\%$  & 0.1629 & -42.9$\%$\\
\rowcolor{Gray}
Shared.F 		 	&& 0.291 	& 0.814 	& 1.089	&& 0.0022 & -99.3$\%$ & 0.1563 & -2.3$\%$\\		
Shared.F*			&& 0.943 	& 0.969 	& 1.015	&& 0.0022 & -& 0.1563 & -\\		
\rowcolor{Gray}
Ind			 	&& 1 	& 1	 	& 1		&& 0 	 & -0.7$\%$  & 0.1553 & -0.4$\%$\\		\hline												
\end{tabular}
} %%%
\end{table}

\end{example}

In the following subsection, we heuristically explain how the PFOS can fully resolve the double-counting problem and provide numerical evidence supporting our heuristic.

\subsection{Double-Counting Problem and Full Optimization}
In the previous subsection, we showed that the PFOS can improve the HMSE to some extent; this can be expected because the PFOS optimizes the HMSE over the {\it a posteriori} as well as {\it a priori} rate functions, unlike the PPOS, where the {\it a priori} rate is predetermined. Following the same logic, we argue that the PFOS can improve the double-counting problem to some extent, depending on how the multiplicative assumption in \eqref{eq.aa.1} deviates from the BM relativity
function in the BMS with the individualized BM relativity table.
However, we cannot expect the PFOS to fully resolve the double-counting problem as shown in Example \ref{ex.3}. This subsection explains that we can expect surprising success of the PFOS compared to the PPOS in resolving the double-counting problem, and that the key reason for the success is
that the PFOS enables the {\it a posteriori} risk classification of the BMS with a shared BM relativity table to be involved with {\it a priori} risk characteristics.

First, recall that Remark \ref{rem.3} in Section \ref{sec.3.2} explains the double-counting problem in the PPOS as follows:

\smallskip
While the policyholder's {\it a posteriori} rate in the PPOS is a function of claim histories only, this does not mean that the {\it a posteriori} rate in the PPOS is not affected by the {\it a priori} rate. Actually, the policyholder's {\it a posteriori} rate is essentially affected by the {\it a priori} rate, because the claim histories are affected by the {\it a priori} rate.
Because the claim histories are affected by the {\it a priori} rate, the {\it  a priori} rate has to be used in the {\it a posteriori} risk classification to remove the {\it a priori} rate effect on the {\it a posteriori} risk classification process. This is how the POI and the Bayesian premium resolve the double-counting problem.
\smallskip

However, the {\it BMS with a shared BM relativity table} does not allow for the {\it  a posteriori} risk classification to be a
function of the {\it a priori} rate, since it is a function of the policyholder's BM level only.
The remedy is to tune the predetermined {\it  a priori} rate function instead. The PFOS allows for freely choosing the {\it a priori} rate function, but the newly chosen {\it a priori} rate function is in fact really the multiplication of the predetermined {\it a priori} rate $\lambda_\kappa$ and the correction for the {\it a posteriori} rate function, as shown in \eqref{eq.b5}.
Thus, in the alternative representation of the PFOS in \eqref{eq.b5}, the PFOS was also interpreted as the {\it BMS with individualized BM relativity table} and restriction in \eqref{eq.aa.1} on BM relativity.
While the previous section showed that we expect only limited improvement in the double-counting problem depending on how the multiplicative restriction in \eqref{eq.aa.1} deviates from the BM relativity function in the POI, the inclusion of the {\it a priori} rate in the {\it a posteriori} risk classification plays the crucial rule of removing the {\it a priori} rate from the policyholder's BM level in the {\it  a posteriori} risk classification process and resolving the double-penalty problem.
In fact, the following lemma explains that by merely allowing the {\it a priori} rate to participate in determining the {\it a posteriori} risk classification, we can fully resolve the double-penalty problem in
the {\it BMS with a shared BM relativity table}, regardless of the BM relativity table given.

\begin{lemma}\label{lem.2}
  For any given positive BM relativities
  \[
  \gamma_{[{\rm Shared}]}(1), \cdots, \gamma_{[{\rm Shared}]}(z),
  \]
  consider the premium in \eqref{eq.5} with the choice of the {\it a priori} rate function given by
  \begin{equation}\label{eq.lem.2}
  \lambda_{[{\rm Shared}]}(\boldsymbol{x}_\kappa)=\frac{\lambda_\kappa}{\E{\gamma_{[{\rm Shared}]}(L)\big\vert \Lambda=\lambda_\kappa}}, \quad \kappa=1, \cdots, \mathcal{K}.
  \end{equation}
  Then, we have
  \begin{equation}\label{eq.lem.1}
  \E{ M_{[{\rm Shared}]}\left(\Lambda, L\right) \big\vert \Lambda= \lambda_\kappa}=\lambda_\kappa
  \end{equation}
  and
  \[
  {\rm FIX}\left( M_{[{\rm Shared}]}\right)=0.
  \]
\end{lemma}
\begin{proof}
  The proof of \eqref{eq.lem.1} is trivial from the definition in \eqref{eq.lem.2}. Then, \eqref{eq.lem.1} concludes that
    \[
  {\rm FIX}\left( M_{[{\rm Shared}]}\right)=0.
  \]
\end{proof}
Now, through a numerical experiment, we can confirm that the main causes of double-counting problem in the BMS is the predetermined {\it  a priori} premium and does not allow for correcting its unfair {\it a posteriori} premium. The double-counting problem can be significantly improved by just allowing for the {\it a priori} rate function to affect the optimization process.

\begin{example}[Continued from Example \ref{ex.1}] \label{ex.4}
We consider a numerical study with the premium obtained from full optimization. As shown in Algorithm \ref{algo.a.1} in the appendix, the premium is obtained from full optimization by applying the coordinate descent algorithm, where the BM relativity and the {\it a priori} rate function are iteratively updated. Specifically, we begin with the case of no {\it a posteriori} risk classification,
\begin{equation}\label{eq.b2}
\left(\gamma_{[{\rm Shared.F}]}^{[0]}, \xi_{[{\rm Shared.F}]}^{[0]}\right):=\left(\tilde{\gamma}_{[{\rm Shared.No}]}, \tilde{\xi}_{[{\rm Shared.No}]}\right),
\end{equation}
which is then iteratively updated in the order
\begin{equation*}%\label{eq.b1}
\begin{aligned}
&\left(\gamma_{[{\rm Shared.F}]}^{[0]}, \xi_{[{\rm Shared.F}]}^{[0]}\right) \rightarrow
\left(\gamma_{[{\rm Shared.F}]}^{[1]}, \xi_{[{\rm Shared.F}]}^{[0]}\right) \rightarrow
\left(\gamma_{[{\rm Shared.F}]}^{[1]}, \xi_{[{\rm Shared.F}]}^{[1]}\right) \rightarrow\\
&\quad\quad\quad\quad\left(\gamma_{[{\rm Shared.F}]}^{[2]}, \xi_{[{\rm Shared.F}]}^{[1]}\right) \rightarrow
 \cdots \rightarrow
\left(\gamma_{[{\rm Shared.F}]}^{[m]}, \xi_{[{\rm Shared.F}]}^{[m]}\right)
\end{aligned}
\end{equation*}
until the termination condition is satisfied. For convenience, we name each relativity pair and the {\it a priori} rate in sequential order;
for example, we name
\[
\left(\gamma_{[{\rm Shared.F}]}^{[0]}, \xi_{[{\rm Shared.F}]}^{[0]}\right)
\]
as the first pair, and
\[
\left(\gamma_{[{\rm Shared.F}]}^{[1]}, \xi_{[{\rm Shared.F}]}^{[0]}\right)
\]
as the second pair. Note that the first pair corresponds to the relativity and the {\it a priori} rate of PNO  as shown in \eqref{eq.b2}, and the second pair corresponds to the relativity and the {\it a priori} rate of  PPOS
\[
\left(\gamma_{[{\rm Shared.F}]}^{[1]}, \xi_{[{\rm Shared.F}]}^{[0]}\right) =
\left(\tilde{\gamma}_{[{\rm Shared.P}]}, \tilde{\xi}_{[{\rm Shared.P}]}\right),
\]
which is implied by Lemma \ref{lem.a2} in the appendix.

In the numerical study, we calculate the HMSE and FIX for each iterated pair, to find the updated step that most significantly contributed to the HMSE or FIX.
The results are presented in Table \ref{tab.ex.4} and Figure \ref{fig.ex.4}, where all the scenarios seem to share similar patterns in the increase and decrease of HMSE and FIX.
The first pair has the largest HMSE but no double penalty with zero FIX.
While the second pair accomplished a major reduction in HMSE, it created the double-penalty problem with a large FIX.
However, the third pair accomplished a major reduction in FIX and some reduction in HMSE. After the third pair, both the HMSE and FIX seem stabilize. An important implication of this numerical study is that the major improvement in terms of double penalty comes from the first participation of the {\it a priori} rate in the optimization problem.

\begin{table}[h!]
\caption{Tracing the HMSE and FIX along the coordinate descent algorithm under various scenarios}\label{tab.ex.4}
\centering
\resizebox{\textwidth}{!}{%%% resize table
\begin{tabular}{l l g c g c g c g c}
  \hline
  \multirow{2}{*}{} & Relativity ftn & $\gamma_{[{\rm Shared.F}]}^{[0]}$ & $\gamma_{[{\rm Shared.F}]}^{[1]}$ & $\gamma_{[{\rm Shared.F}]}^{[1]}$ & $\gamma_{[{\rm Shared.F}]}^{[2]}$
  & $\gamma_{[{\rm Shared.F}]}^{[2]}$ & $\cdots$ & $\gamma_{[{\rm Shared.F}]}^{[m]}$ \\
  & Priori ftn    & $\xi_{[{\rm Shared.F}]}^{[0]}$ & $\xi_{[{\rm Shared.F}]}^{[0]}$ & $\xi_{[{\rm Shared.F}]}^{[1]}$  & $\xi_{[{\rm Shared.F}]}^{[1]}$ & $\xi_{[{\rm Shared.F}]}^{[2]}$
  & $\cdots$ &  $\xi_{[{\rm Shared.F}]}^{[m]}$ \\
   \hline
  \multirow{2}{*}{Scenario I} 	& FIX 	& 0.0000& 0.3075& 0.0013& 0.0047& 0.0021& $\cdots$ & 0.0022 \\
  							& HMSE 	& 0.2853& 0.1629& 0.1564& 0.1563& 0.1563& $\cdots$ & 0.1563 \\
  \hline
  \multirow{2}{*}{Scenario II} 	& FIX 	& 0.0000& 0.0182& 0.0004& 0.0004& 0.0004& $\cdots$ & 0.0004 \\
  							& HMSE 	& 0.2053& 0.0988& 0.0970& 0.0970& 0.0970&$\cdots$ & 0.0970 \\
  \hline
  \multirow{2}{*}{Scenario III} & FIX 	& 0.0000& 0.0626 & 0.0021 & 0.0023 & 0.0023 & $\cdots$ & 0.0023 \\
  							& HMSE 	& 0.8853& 0.5598 & 0.5476 & 0.5476 & 0.5476 & $\cdots$ & 0.5476 \\
  \hline
  \multirow{2}{*}{Scenario IV}  & FIX 	& 0.0000 & 0.4686 & 0.0075 & 0.0113 & 0.0093 &  $\cdots$ & 0.0096 \\
  							& HMSE 	& 0.3973& 0.2560 & 0.2474 & 0.2473 & 0.2473 &  $\cdots$ & 0.2473 \\
  \hline
\end{tabular}
}

\end{table}

\begin{figure}[p]
     \centering
    \subfloat[Trace plot of the HMSE and FIX under Scenario I]
    {%
      \includegraphics[width=0.65\textwidth]{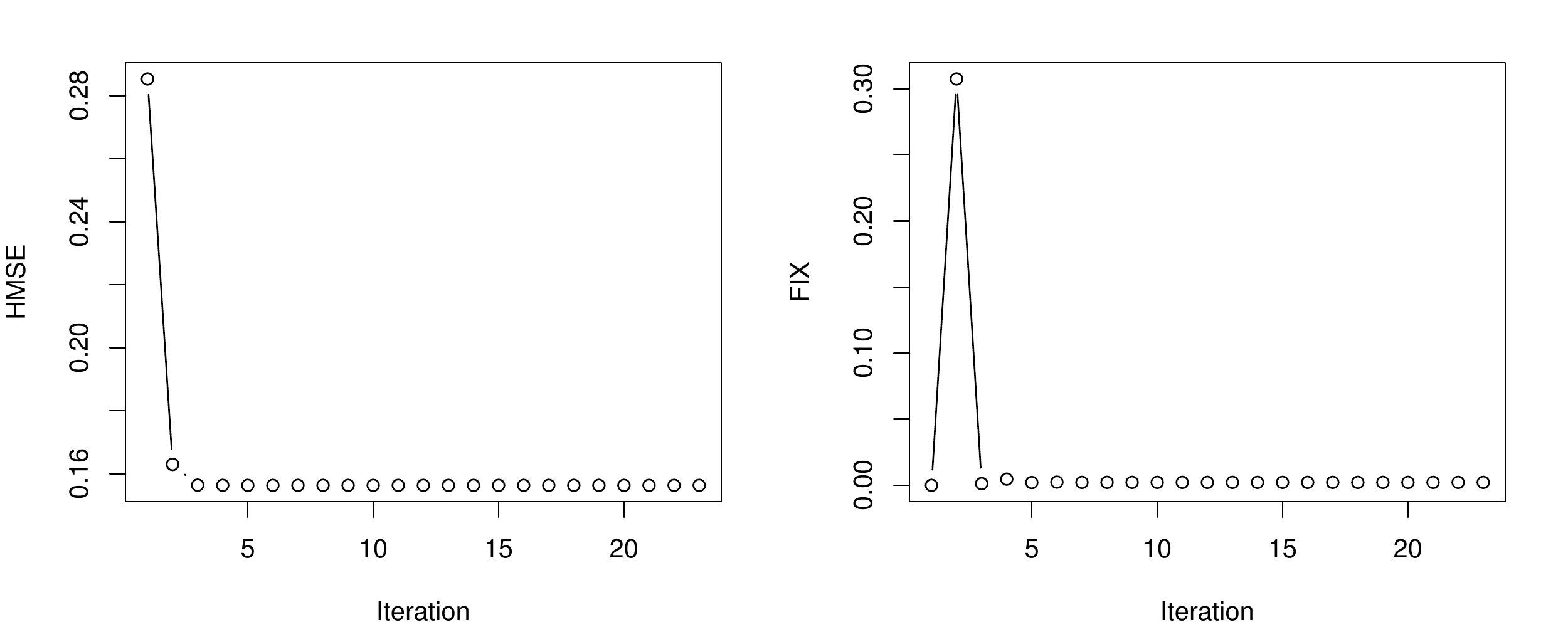}
    }

    \subfloat[Trace plot of the HMSE and FIX under Scenario II]
    {%
      \includegraphics[width=0.65\textwidth]{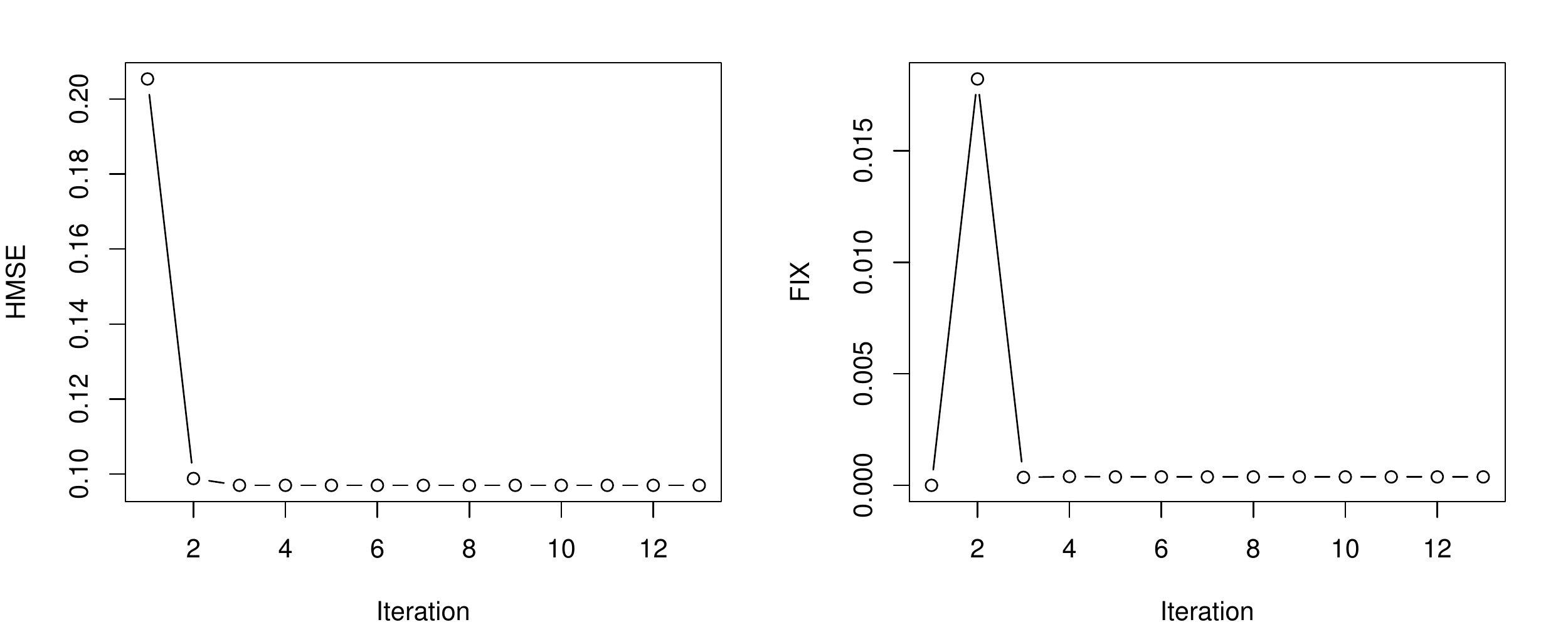}
    }\\

      \subfloat[Trace plot of the HMSE and FIX under Scenario III]
      {%
      \includegraphics[width=0.65\textwidth]{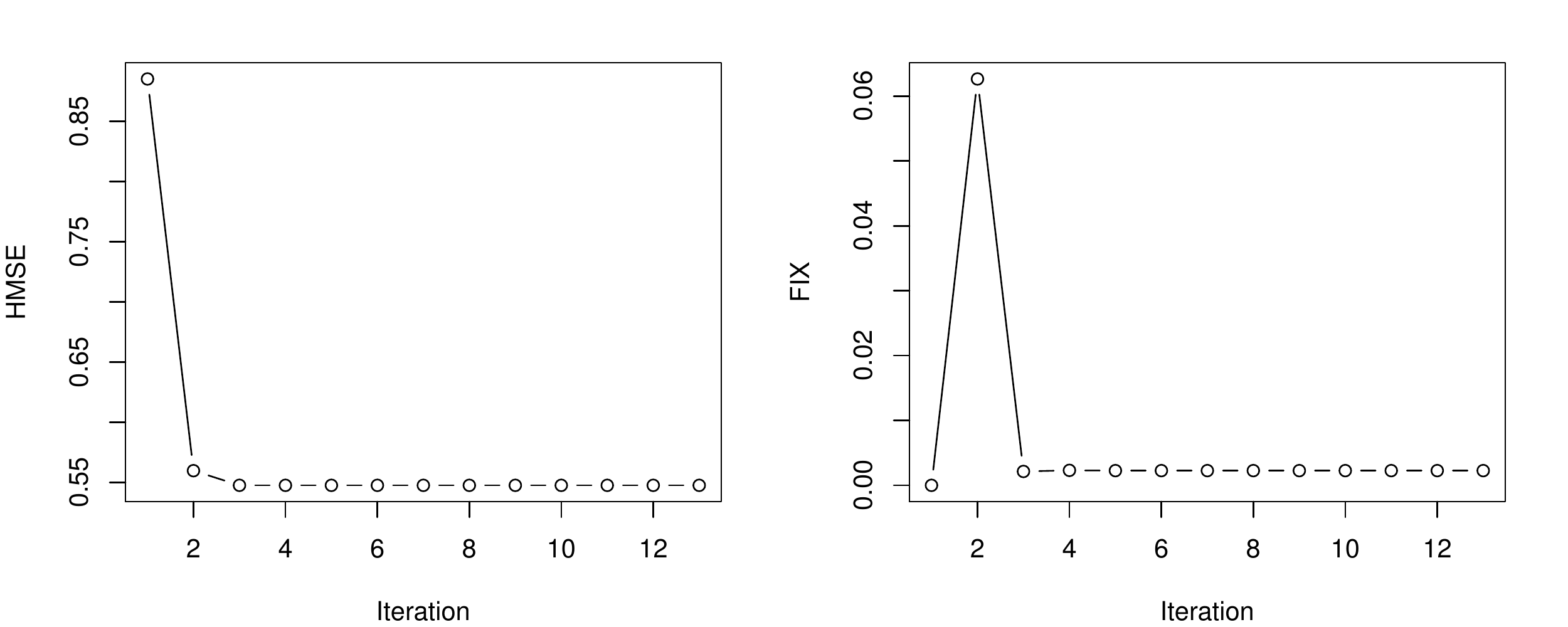}
    }\\

    \subfloat[Trace plot of the HMSE and FIX under Scenario IV]
    {%
      \includegraphics[width=0.65\textwidth]{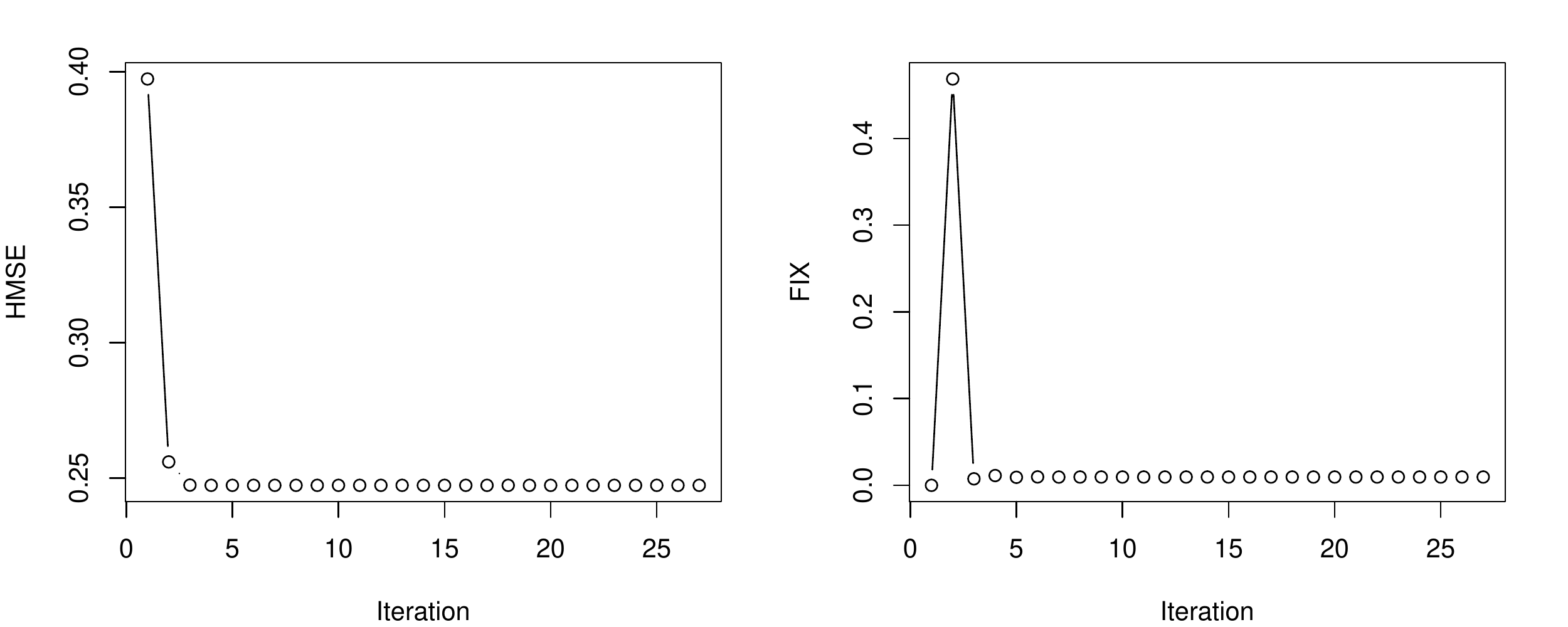}
    }
\caption{Trace plots of the HMSE and FIX along the coordinate descent algorithm under various settings until the termination condition is satisfied}
\label{fig.ex.4}
  \end{figure}

\end{example}

\section{Real Data Analysis}\label{sec.5}

To examine to what extent the double-counting problem can be resolved with the proposed full optimization process in an actual BMS,
we use a collision coverage dataset for new and old vehicles from the Wisconsin Local Government Property Insurance Fund (LGPIF) in \citet{Frees4}.
Collision coverage provides for coverage for impact of a vehicle with an object, impact of a vehicle with an attached vehicle, or overturn of a vehicle.
In our analysis,
we use the longitudinal data over the policy years 2006 to 2010 for a total of 497 governmental entities.
We have two categorical variables, the entity type with six levels and coverage at three levels, as shown
in Table \ref{tab.x} (similar to Table C.6 in Oh et al. (2019)).

Specifically, for the frequency random effect model,
we consider
\[
  N_{i,t}\big\vert \left(\Lambda_{i}, \Theta_i\right)=\left(\lambda_{\kappa}, \theta_i\right) \iid
        {\rm Pois}(\lambda_{\kappa}\theta_i)
\quad \hbox{with} \quad \lambda_\kappa=\exp(\boldsymbol{x}_{\kappa}\boldsymbol{\beta}),
\]
where $\boldsymbol{x}_\kappa$ is a vector of the two categorical variables in Table \ref{tab.x} for the $\kappa$-th risk group, and $\boldsymbol{\beta}$ is the corresponding parameters.
For the random effect, we assume a gamma distribution with the mean of 1 and dispersion parameter $\psi$, that is,
$\Theta_i \sim {\rm Gamma}(1, \psi)$.

\subsection{Estimation Result}

Table \ref{est} summarizes the estimation results of the frequency random effect model.
Note that the asterisk * indicates that the parameters are significant at the 0.05. % level.
The results are quite intuitive.
For example, the county entity type is expected to incur more accidents
than any other entity type with the same amount of coverage.
The risk group with higher coverage amount tends to incur more accidents for the same entity type.

From the estimates in Table \ref{est},
we find that  Table \ref{bm} presents the results of various relativities under the -1/+2 transition rule; that is, the shared relativity table from partial optimization
$\tilde{\gamma}_{[{\rm Shared.P}]}$, shared relativity table from full optimization $\tilde{\gamma}_{[{\rm Shared.F}]}$, pure relativity table from full optimization $\tilde{\gamma}^*_{[{\rm Shared.F}]}$, and individualized relativity table $\tilde{\gamma}_{[{\rm Ind}]}$.
In our data analysis, we set 18 different risk class groups
by combining two categorical variables, the entity type and coverage group,
so that
the BMS with individualized BM relativity (shared.F$^*$ and Ind) has 18 different BM relativity tables for the 18 risk class groups.
To compare the premiums under various methods, we report the FIX and HMSE as measuring the double-counting problem and quality of predictor, respectively. The {\it a priori} rate functions are summarized in Table \ref{pri.rate.ftn}.
Note that the {\it a priori} rate functions $\tilde{\xi}_{[{\rm Shared.P}]}$ and $\tilde{\xi}_{[{\rm Ind}]}$ are given by the estimated {\it a priori} rate $\hat{\lambda}_\kappa=\exp(\boldsymbol{x}_{\kappa}\boldsymbol{\hat{\beta}})$ obtained from the fixed effect estimates in Table \ref{est}.

From Table \ref{bm}, we observe the impact of using the proposed optimization process on the BM premium.
We can provide explanations similar to those in Example \ref{ex.3}.
 The PFOS seems to mitigate the double-counting problem up to a level similar to the individualized BMS level with a much smaller FIX than that from the PPOS, and a bit smaller
  HMSE value, suggesting that the proposed optimization process leads to a slight increase in prediction power.

\begin{table}[h!]
\caption{Estimation result}
\centering
\begin{tabular}{ l r r r r r r r  }
 \hline
parameter& estimate & std.error & \multicolumn{1}{c}{p-value} \\
 \hline
 \multicolumn{4}{l}{ {\bf Fixed effect}} \\
(Intercept)		&	-3.219&	0.402&	$<$.0001 &* \\%-8.00& * \\
TypeCity			& 	0.921&	0.424&	0.0298 	&* \\%2.17&*	\\
TypeCounty		&	2.066&	0.420&	$<$.0001 &* \\%4.92&*	\\
TypeSchool		&	0.748&	0.392&	0.0565	&  \\%1.91&	\\
TypeTown			&	-1.194&	0.470&	0.0111 	&* \\%-2.54&*	\\
TypeVillage		&	0.229&	0.409&	0.5760	&  \\%0.56&	\\
catCoverage2		&	1.439&	0.206&	$<$.0001 &* \\%6.97&*	\\
catCoverage3		&	2.344&	0.224&	$<$.0001 &* \\%10.47&* \\
\hline
 \multicolumn{4}{l}{ {\bf Random effect}} \\
 $\psi$			& 0.782& \\
 \hline
\end{tabular}
\label{est}
\end{table}

\begin{table}
\caption{Various BM relativities and distribution of $L$ under the -1/+2 transition rule}
\centering
\resizebox{\linewidth}{!}{%%% resize table
\begin{tabular}{ l l*{25}{c}}
 \hline
 			 &Risk       & \multicolumn{10}{c}{Level $\ell$ } \\ \cline{3-12}
 	 &group $\kappa$	& 1& 2& 3& 4& 5& 6& 7& 8& 9& 10 &&  FIX & HMSE\\
\hline
\rowcolor{Gray}
$\tilde{\gamma}_{[{\rm Shared.P}]}$&
&	0.16 	&	0.22 	&	0.23 	&	0.27 	&	0.29 	&	0.32 	&	0.36 	&	0.43 	&	0.55 	&	1.23 	&&	0.511 	&	0.871 	\\
%$\tilde{\gamma}_{\rm [Shared.F]}$&
%&	0.46 	&	0.71 	&	0.75 	&	0.92 	&	1.00 	&	1.15 	&	1.29 	&	1.51 	&	1.95 	&	4.32 	&&	0.017 	&	0.844 	\\
$\tilde{\gamma}_{\rm [Shared.F]}$& 	&	0.13 	&	0.20 	&	0.22 	&	0.26 	&	0.29 	&	0.33 	&	0.37 	&	0.43 	&	0.56 	&	1.23 	&&	0.017 	&	0.844 	\\
\rowcolor{Gray}
$\tilde{\gamma}_{\rm [Shared.F]}^*$ & 1	&	0.95 	&	1.47 	&	1.55 	&	1.89 	&	2.05 	&	2.35 	&	2.64 	&	3.11 	&	4.00 	&	8.87 	&&	0.017 	&	0.844 	\\
\rowcolor{Gray}
&	2	&	0.48 	&	0.74 	&	0.78 	&	0.96 	&	1.04 	&	1.19 	&	1.34 	&	1.57 	&	2.02 	&	4.49 	&&		&		\\
\rowcolor{Gray}
&	3	&	0.26 	&	0.40 	&	0.42 	&	0.51 	&	0.56 	&	0.64 	&	0.72 	&	0.84 	&	1.08 	&	2.40 	&&		&		\\
\rowcolor{Gray}
&	4	&	0.72 	&	1.10 	&	1.16 	&	1.42 	&	1.54 	&	1.77 	&	1.99 	&	2.34 	&	3.01 	&	6.68 	&&		&		\\
\rowcolor{Gray}
&	5	&	0.26 	&	0.39 	&	0.42 	&	0.51 	&	0.55 	&	0.63 	&	0.71 	&	0.83 	&	1.07 	&	2.38 	&&		&		\\
\rowcolor{Gray}
&	6	&	0.17 	&	0.26 	&	0.28 	&	0.34 	&	0.37 	&	0.42 	&	0.48 	&	0.56 	&	0.72 	&	1.60 	&&		&		\\
\rowcolor{Gray}
&	7	&	0.31 	&	0.47 	&	0.49 	&	0.60 	&	0.66 	&	0.75 	&	0.85 	&	0.99 	&	1.28 	&	2.84 	&&		&		\\
\rowcolor{Gray}
&	8	&	0.16 	&	0.24 	&	0.26 	&	0.31 	&	0.34 	&	0.39 	&	0.44 	&	0.52 	&	0.67 	&	1.48 	&&		&		\\
\rowcolor{Gray}
&	9	&	0.13 	&	0.20 	&	0.21 	&	0.25 	&	0.28 	&	0.32 	&	0.36 	&	0.42 	&	0.54 	&	1.20 	&&		&		\\
\rowcolor{Gray}
&	10	&	0.79 	&	1.22 	&	1.29 	&	1.57 	&	1.71 	&	1.96 	&	2.20 	&	2.59 	&	3.33 	&	7.39 	&&		&		\\
\rowcolor{Gray}
&	11	&	0.28 	&	0.43 	&	0.46 	&	0.56 	&	0.61 	&	0.70 	&	0.78 	&	0.92 	&	1.19 	&	2.63 	&&		&		\\
\rowcolor{Gray}
&	12	&	0.18 	&	0.28 	&	0.30 	&	0.36 	&	0.39 	&	0.45 	&	0.51 	&	0.59 	&	0.76 	&	1.70 	&&		&		\\
\rowcolor{Gray}
&	13	&	0.99 	&	1.52 	&	1.60 	&	1.96 	&	2.13 	&	2.44 	&	2.74 	&	3.22 	&	4.14 	&	9.19 	&&		&		\\
\rowcolor{Gray}
&	14	&	0.93 	&	1.43 	&	1.51 	&	1.84 	&	2.00 	&	2.29 	&	2.57 	&	3.03 	&	3.89 	&	8.64 	&&		&		\\
\rowcolor{Gray}
&	15	&	0.61 	&	0.93 	&	0.99 	&	1.20 	&	1.31 	&	1.50 	&	1.69 	&	1.98 	&	2.55 	&	5.66 	&&		&		\\
\rowcolor{Gray}
&	16	&	0.93 	&	1.43 	&	1.51 	&	1.84 	&	2.00 	&	2.29 	&	2.58 	&	3.03 	&	3.90 	&	8.66 	&&		&		\\
\rowcolor{Gray}
&	17	&	0.40 	&	0.62 	&	0.65 	&	0.80 	&	0.87 	&	0.99 	&	1.12 	&	1.32 	&	1.69 	&	3.76 	&&		&		\\
\rowcolor{Gray}
&	18	&	0.23 	&	0.35 	&	0.37 	&	0.45 	&	0.49 	&	0.56 	&	0.63 	&	0.75 	&	0.96 	&	2.13 	&&		&		\\

$\tilde{\gamma}_{\rm [Ind]}$& 	
1&	0.92 	&	1.65 	&	1.69 	&	2.39 	&	2.52 	&	3.16 	&	3.40 	&	3.96 	&	4.29 	&	4.80 	&&	0.000 	&	0.834 	\\
&2&	0.62 	&	1.03 	&	1.09 	&	1.44 	&	1.57 	&	1.86 	&	2.07 	&	2.35 	&	2.63 	&	2.97 	&&		&		\\
&3&	0.35 	&	0.55 	&	0.59 	&	0.76 	&	0.85 	&	1.02 	&	1.17 	&	1.38 	&	1.66 	&	2.07 	&&		&		\\
&4&	0.77 	&	1.32 	&	1.39 	&	1.87 	&	2.02 	&	2.43 	&	2.67 	&	3.03 	&	3.32 	&	3.68 	&&		&		\\
&5&	0.34 	&	0.55 	&	0.59 	&	0.75 	&	0.84 	&	1.00 	&	1.15 	&	1.37 	&	1.64 	&	2.06 	&&		&		\\
&6&	0.16 	&	0.25 	&	0.27 	&	0.35 	&	0.40 	&	0.48 	&	0.58 	&	0.73 	&	0.98 	&	1.55 	&&		&		\\
&7&	0.42 	&	0.68 	&	0.73 	&	0.94 	&	1.05 	&	1.24 	&	1.41 	&	1.64 	&	1.92 	&	2.30 	&&		&		\\
&8&	0.13 	&	0.20 	&	0.22 	&	0.28 	&	0.32 	&	0.39 	&	0.47 	&	0.60 	&	0.84 	&	1.45 	&&		&		\\
&9&	0.05 	&	0.09 	&	0.09 	&	0.12 	&	0.14 	&	0.17 	&	0.21 	&	0.28 	&	0.43 	&	1.21 	&&		&		\\
&10&	0.81 	&	1.41 	&	1.47 	&	2.01 	&	2.16 	&	2.61 	&	2.86 	&	3.25 	&	3.55 	&	3.93 	&&		&		\\
&11&	0.39 	&	0.62 	&	0.67 	&	0.86 	&	0.96 	&	1.14 	&	1.30 	&	1.52 	&	1.80 	&	2.19 	&&		&		\\
&12&	0.19 	&	0.29 	&	0.32 	&	0.41 	&	0.46 	&	0.56 	&	0.66 	&	0.83 	&	1.08 	&	1.62 	&&		&		\\
&13&	0.98 	&	1.75 	&	1.77 	&	2.54 	&	2.59 	&	3.34 	&	3.45 	&	4.17 	&	4.33 	&	5.02 	&&		&		\\
&14&	0.89 	&	1.59 	&	1.65 	&	2.30 	&	2.44 	&	3.03 	&	3.28 	&	3.79 	&	4.12 	&	4.58 	&&		&		\\
&15&	0.71 	&	1.20 	&	1.27 	&	1.68 	&	1.83 	&	2.18 	&	2.40 	&	2.73 	&	3.01 	&	3.36 	&&		&		\\
&16&	0.90 	&	1.60 	&	1.65 	&	2.31 	&	2.45 	&	3.04 	&	3.29 	&	3.80 	&	4.13 	&	4.59 	&&		&		\\
&17&	0.55 	&	0.90 	&	0.96 	&	1.25 	&	1.37 	&	1.62 	&	1.81 	&	2.08 	&	2.35 	&	2.70 	&&		&		\\
&18&	0.29 	&	0.46 	&	0.49 	&	0.64 	&	0.71 	&	0.85 	&	0.99 	&	1.19 	&	1.46 	&	1.91 	&&		&		\\

\hline\hline	
$\P{L=\ell}$ && 0.534& 0.044& 0.052& 0.023& 0.024& 0.020& 0.024& 0.031& 0.055& 0.191 \\		
\hline												
\end{tabular}
} %%%
\label{bm}
\end{table}

\begin{table}
\caption{The {\it a priori} rate function }\label{pri.rate.ftn}
\centering
\resizebox{\linewidth}{!}{%%% resize table
%\resizebox{!}{\textheight}{%%% resize table
\begin{tabular}{ l l*{25}{c}}
 \hline
 	 Risk group 	& 1& 2& 3& 4& 5& 6& 7& 8& 9& 10 & 11 & 12 & 13 & 14 & 15 & 16 & 17 & 18\\
\hline
$\tilde{\xi}_{[{\rm Shared.P}]}$ &
0.04 	&	0.17 	&	0.42 	&	0.10 	&	0.42 	&	1.05 	&	0.32 	&
1.33 	&	3.29 	&	0.08 	&	0.36 	&	0.88 	&	0.01 	&	0.05 	&
0.13 	&	0.05 	&	0.21 	&	0.52 	\\						
$\tilde{\xi}_{[{\rm Shared.F}]}$&
0.29 	&	0.61 	&	0.81 	&	0.54 	&	0.82 	&	1.35 	&	0.72 	&
1.59 	&	3.18 	&	0.50 	&	0.76 	&	1.21 	&	0.09 	&	0.36 	&
0.58 	&	0.35 	&	0.64 	&	0.90 	\\		
$\tilde{\xi}_{[{\rm Ind}]}$	 &
0.04 	&	0.17 	&	0.42 	&	0.10 	&	0.42 	&	1.05 	&	0.32 	&
1.33 	&	3.29 	&	0.08 	&	0.36 	&	0.88 	&	0.01 	&	0.05 	&
0.13 	&	0.05 	&	0.21 	&	0.52 	\\									
\hline												
\end{tabular}
} %%%
\end{table}

\section{Conclusion}
In this study, we revisit the double-counting problem in the BMS. We find that the problem originated from the inefficiency in optimization, with only the {\it a posteriori} function contributing. While the double-counting problem can be fully resolved by allowing individualized BM relativity tables for each policyholder, which we call the POI, such a BMS is complicated and unsuitable for effective communication between the policyholders and insurers.
By allowing for both the {\it a priori} and {\it a posteriori} rate functions to participate in the optimization process, which we call the PFOS, we show that the double-counting problem can be virtually resolved without complicating the traditional BMS. In our analysis, we found no significant difference between the PFOS and POI in terms of double counting (via FIX) or prediction ability (via HMSE). A future study may find it interesting to compare the PFOS and POI under various circumstances and identify the conditions under which two BMS show significantly different patterns both in terms of double counting and prediction ability.

\section*{Acknowledgements}
 Jae Youn Ahn was supported by a National Research Foundation of Korea (NRF) grant funded
by the Korean Government (NRF-2017R1D1A1B03032318).
Rosy Oh was supported by Basic Science Research Program through the National Research Foundation of Korea(NRF) funded by the Ministry of Education (Grant No. 2019R1A6A1A11051177).

\newpage

%\section*{References}
\bibliographystyle{apalike}
\bibliography{Bib_Oh}

\newpage

\appendix
\section{Optimization Algorithm}

From the optimization perspective, the two problems in \eqref{eq.a.7} and \eqref{eq.a.70} are equivalent, as shown in \eqref{eq.a3}. Specifically, the problem in \eqref{eq.a.70} can be obtained by solving the problem $(\tilde{\xi}_{[{\rm Shared.F}]}, \tilde{\gamma}_{[{\rm Shared.F}]})$ of \eqref{eq.a.7} by finding a constant $c>0$ satisfying
\[
c\, \tilde{\gamma}_{[{\rm Shared.F}]} \left( \left\lfloor \frac{z}{2} \right\rfloor\right)=q;
\]
in this case, the problem in \eqref{eq.a.70} is obtained as
\[
\left(\frac{1}{c}\tilde{\xi}_{[{\rm Shared.F}]}, c\, \tilde{\gamma}_{[{\rm Shared.F}]}\right).
\]
Thus, for the optimization of \eqref{eq.a.70}, we need to solve  \eqref{eq.a.7}. In the following algorithm, we adopt the {\it coordinate descent algorithm}, where the BM relativity and {\it a priori} rate functions are iteratively optimized.

\begin{algo}\label{algo.a.1}
  Under the {\it BMS with shared BM relativity table}, we can solve \eqref{eq.a.70} through iterative calculation using the following coordinate descent algorithm:
  \begin{enumerate}
    \item[i.] Initial step:
    \begin{itemize}
      \item $m \leftarrow 0$
      \item Define functions ${\gamma}_{[{\rm Shared}]}^{[m]}:\{1, \cdots, z\}\mapsto \Real_+$ and $\xi_{[{\rm Shared}]}^{[m]}: \{1, \cdots, \mathcal{K}\}\mapsto \Real_+$ as
   \[
  \gamma_{[{\rm Shared}]}^{[m]}(\ell):=\frac{\E{\Lambda^2 \Theta\big\vert {L}=\ell}}{\E{ \Lambda^2 \big\vert {L}=\ell}}, \quad \ell=1, \cdots, {z}
  \]
  and
  \[
  \xi_{[{\rm Shared}]}^{[m]}(\lambda_{\kappa}):=\lambda_{\kappa}, \quad \kappa=1, \cdots, \mathcal{K},
  \]
respectively.
    \end{itemize}

 \item[ii.] Repeat step:
 \begin{itemize}
   \item  For the given ${\gamma}_{[\rm Shared]}^{[m]}:\{1, \cdots, z\}\mapsto \Real_+$, define
     \begin{equation}\label{eq.a5}
  {\gamma}_{[{\rm Shared}]}^{[m+1]}(\ell):=\frac{\E{\xi_{[{\rm Shared}]}^{[m]}(\Lambda)\,\Lambda \Theta\big\vert {L}=\ell}}{\E{ \left( \xi_{[{\rm Shared}]}^{[m]}(\Lambda) \right)^2\big\vert {L}=\ell}}, \quad \ell=1, \cdots, {z},
  \end{equation}
  and for the given $\xi_{[\rm Shared]}^{[m]}: \{1, \cdots, \mathcal{K}\}\mapsto \Real_+$, define
 \begin{equation}\label{eq.a6}
  \xi_{[{\rm Shared}]}^{[m+1]}(\lambda_{\kappa}):=\lambda_{\kappa} \frac{\E{{\gamma}_{[\rm Shared]}^{[m+1]}({L})\, \Theta\big\vert \Lambda=\lambda_{\kappa}}}{\E{ \left( {\gamma}_{[\rm Shared]}^{[m+1]}({L}) \right)^2\big\vert \Lambda=\lambda_{\kappa}}}, \quad \kappa=1, \cdots, \mathcal{K}.
  \end{equation}
  \item $m \leftarrow m+1$

 \end{itemize}
\item[iii.] Termination and normalization
 \begin{itemize}
   \item Repeat step ii until the termination test is satisfied.
   \item Calculate the normalizing constant $c$ such that $c\,{\gamma}_{[Shared]}^{[m]}(\left\lfloor \frac{z}{2}\right\rfloor)=q$.
   \item The pair of functions
   $$c\,\gamma_{[{\rm Shared}]}^{[m]}\quad\hbox{and}\quad\frac{1}{c}\xi_{[{\rm Shared}]}^{[m]}$$
    is approximately the pair of functions $\tilde{\gamma}_{[{\rm Shared.F}]}$ and $\tilde{\xi}_{[{\rm Shared.F}]}$, which is the solution of \eqref{eq.a.70}.
 \end{itemize}

  \end{enumerate}
\end{algo}

The iterative updates of Step ii in Algorithm \ref{algo.a.1} come from Lemma \ref{lem.a1}, with the specific updating formula given in Lemma \ref{lem.a2}.

\begin{lemma}\label{lem.a1}
Algorithm \ref{algo.a.1} gives the following auxiliary results.
\begin{enumerate}
  \item[i.] For the given {\it a priori} rate function $\xi_{[\rm Shared]}^{[m]}$,
the BM relativity function ${\gamma}_{[{\rm Shared}]}^{[m+1]}$ in step ii of Algorithm \ref{algo.a.1} provides the solution $\gamma$ of the optimization problem
\[
\min_{ {\gamma} } \E{\left(\Lambda\Theta- \xi_{[\rm Shared]}^{[m]}(\Lambda){\gamma}({L}) \right)^2}.
\]
\item[ii.] For the BM relativity function ${\gamma}_{[{\rm Shared}]}^{[m+1]}$,
the {\it a priori} rate function $\xi_{[\rm Shared]}^{[m+1]}$ in step ii of Algorithm \ref{algo.a.1}
is the solution $\xi$ of the optimization problem
\[
\min_{ {\xi} } \E{\left(\Lambda\Theta- \xi(\Lambda){\gamma}_{[{\rm Shared}]}^{[m+1]}({L}) \right)^2}.
\]
\end{enumerate}
\end{lemma}
\begin{proof}
 The proof of part i is similar to deriving the optimal BM relativity of the {\it premium through partial optimization of the shared BM relativity table}. Finally, the optimization problem in part ii can be written as
\[
\begin{aligned}
\min_{ {\xi} } \E{\left(\Lambda\Theta- \xi(\Lambda){\gamma}_{[{\rm Shared}]}^{[m+1]}({L}) \right)^2}=
\min_{ {\xi} } \sum\limits_{\kappa=1}^{\mathcal{K}}\E{\left(\Lambda\Theta- \xi(\Lambda){\gamma}_{[{\rm Shared}]}^{[m+1]}({L}) \right)^2\big\vert \Lambda=\lambda_\kappa} w_\kappa,
\end{aligned}
\]
which in turn implies that
 \[
  \xi(\lambda_{\kappa})=\lambda_{\kappa} \frac{\E{{\gamma}_{[\rm Shared]}^{[m+1]}({L})\, \Theta\big\vert \Lambda=\lambda_{\kappa}}}{\E{ \left( {\gamma}_{[\rm Shared]}^{[m+1]}({L}) \right)^2\big\vert \Lambda=\lambda_{\kappa}}}
  \]
is the solution of the optimization problem of part ii.
\end{proof}

\begin{lemma}\label{lem.a2}
In step ii of Algorithm \ref{algo.a.1}, the updates in the coordinate algorithm can be calculated as follows.
\begin{enumerate}
  \item[i.] For the given {\it a priori} rate function ${\xi}_{[\rm Shared]}^{[m]}$, the update for the BM relativity function in \eqref{eq.a5} is obtained as
      \[
      {\gamma}_{[{\rm Shared}]}^{[m+1]}(\ell)=
      \frac{
      \sum\limits_{\kappa=1}^{\mathcal{K}}
      \lambda_\kappa  \xi_{[{\rm Shared}]}^{[m]}\left( \lambda_\kappa\right) w_\kappa
      \int\theta\,\pi_\ell\left(\lambda_\kappa\,\theta, \psi\right)g(\theta) {\rm d}\theta
      }
      {
      \sum\limits_{\kappa=1}^{\mathcal{K}}
      \left(\xi_{[{\rm Shared}]}^{[m]}\left( \lambda_\kappa\right)\right)^2 w_\kappa
      \int\theta\,\pi_\ell\left(\lambda_\kappa\,\theta, \psi\right)g(\theta) {\rm d}\theta.
      }
      \]
  \item[ii.] For the given BM relativity function ${\gamma}_{[\rm Shared]}^{[m+1]}$, the update for the {\it a priori} rate function in \eqref{eq.a6} is obtained as
      \[
      {\xi}_{[{\rm Shared}]}^{[m+1]}(\kappa)=
      \lambda_\kappa
      \frac{
      \sum\limits_{\ell=1}^{z}
        \gamma_{[{\rm Shared}]}^{[m+1]}\left( \ell\right)
      \int\theta\,\pi_\ell\left(\lambda_\kappa\,\theta, \psi\right)g(\theta) {\rm d}\theta
      }
      {
      \sum\limits_{\ell=1}^{z}
        \left(\gamma_{[{\rm Shared}]}^{[m+1]}\left( \ell\right)\right)^2
      \int \pi_\ell\left(\lambda_\kappa\,\theta, \psi\right)g(\theta) {\rm d}\theta      }.
      \]
\end{enumerate}
\end{lemma}

\begin{proof}
  The proof of the first part follows from
  \[
  \begin{aligned}
      {\gamma}_{[{\rm Shared}]}^{[m+1]}(\ell)&=
      \frac{
      \sum\limits_{\kappa=1}^{\mathcal{K}}
      \int   \xi_{[{\rm Shared}]}^{[m]}\left( \lambda_\kappa\right)
      \lambda_\kappa\,\theta
      f\left(\lambda_\kappa, \theta |\ell \right){\rm d}\theta
      }
      {
      \sum\limits_{\kappa=1}^{\mathcal{K}}
      \int   \left(\xi_{[{\rm Shared}]}^{[m]}\left( \lambda_\kappa\right) \right)^2
      f\left(\lambda_\kappa, \theta | \ell \right){\rm d}\theta
      }\\
      &=
      \frac{
      \sum\limits_{\kappa=1}^{\mathcal{K}}
      \int
      \xi_{[{\rm Shared}]}^{[m]}\left( \lambda_\kappa\right) \lambda_\kappa\,
      \theta\,\pi_\ell\left(\lambda_\kappa\,\theta, \psi\right)w_\kappa g(\theta) {\rm d}\theta
      }
      {
      \sum\limits_{\kappa=1}^{\mathcal{K}}
      \int
      \left(\xi_{[{\rm Shared}]}^{[m]}\left( \lambda_\kappa\right)\right)^2
      \theta\,\pi_\ell\left(\lambda_\kappa\,\theta, \psi\right)w_\kappa g(\theta) {\rm d}\theta
      }\\
      &=
      \frac{
      \sum\limits_{\kappa=1}^{\mathcal{K}}
      \lambda_\kappa  \xi_{[{\rm Shared}]}^{[m]}\left( \lambda_\kappa\right) w_\kappa
      \int\theta\,\pi_\ell\left(\lambda_\kappa\,\theta, \psi\right)g(\theta) {\rm d}\theta
      }
      {
      \sum\limits_{\kappa=1}^{\mathcal{K}}
      \left(\xi_{[{\rm Shared}]}^{[m]}\left( \lambda_\kappa\right)\right)^2 w_\kappa
      \int\theta\,\pi_\ell\left(\lambda_\kappa\,\theta, \psi\right)g(\theta) {\rm d}\theta
      },
  \end{aligned}
  \]
  where $f\left(\lambda_\kappa, \theta |\ell \right)$ is the density function of $(\Lambda, \Theta)$ at
  $ (\Lambda, \Theta)=(\lambda_\kappa, \theta)$,
  conditional on $L=\ell$. The proof of the second part can be similarly obtained.
\end{proof}

Since a convex and differentiable function is guaranteed to converge to the global minimizer,
Lemma \ref{lem.a1} guarantees the solution in Algorithm \ref{algo.a.1} to converge to the global solution
for the objective function
\begin{equation*}
Q\left({\xi}_{[{\rm Shared}]}(\boldsymbol{x}_1), \cdots, {\xi}_{[{\rm Shared}]}(\boldsymbol{x}_{\kappa}),
{\gamma}_{[{\rm Shared}]}({1}), \cdots, {\gamma}_{[{\rm Shared}]}({z})
\right):=
 \E{\left(\Lambda\Theta- {\xi}_{[{\rm Shared}]}(\Lambda){\gamma}_{[{\rm Shared}]}({L}) \right)^2},
\end{equation*}
which can be easily shown as a convex function and is differentiable for each coordinate.
Refer to \citet{Luo1992} for more details of convergence in the coordinate descent algorithm.

%
%Define
%\begin{equation*}
%Q\left({\xi}_{[{\rm Shared}]}(\boldsymbol{x}_1), \cdots, {\xi}_{[{\rm Shared}]}(\boldsymbol{x}_{\kappa}),
%{\gamma}_{[{\rm Shared}]}({1}), \cdots, {\gamma}_{[{\rm Shared}]}({z})
%\right):=
% \E{\left(\Lambda\Theta- {\xi}_{[{\rm Shared}]}(\Lambda){\gamma}_{[{\rm Shared}]}({L}) \right)^2},
%\end{equation*}
%which can be easily shown as a convex function and is differentiable for each coordinate. Since a convex and differentiable function is guaranteed to converge to the global minimizer,
%Lemma \ref{lem.a1} guarantees the solution in Algorithm \ref{algo.a.1} to converge to the global solution. Refer to \citet{Luo1992} for more details of convergence in the coordinate descent algorithm.
%

\section{Miscellaneous Results}

\begin{lemma}\label{lem.a3}
For the premium $M(\Lambda, L)$ of the BMS, we have the following results:
\begin{enumerate}
  \item[i.] $\Var{\E{\frac{M(\Lambda, L)}{\Lambda}\bigg\vert \Lambda}} =
  \sum_{\kappa} w_{\kappa} \left\{\E{\frac{M(\Lambda, L)}{\Lambda}\bigg\vert \Lambda=\lambda_{\kappa}} - \E{\frac{M(\Lambda, L)}{\Lambda}} \right\}^2
  $.
  \item[ii.] $\E{\Var{\frac{M(\Lambda, L)}{\Lambda}\bigg\vert \Lambda}} =
  \sum_{\kappa} w_{\kappa} \sum_{\ell=1}^z \left\{\frac{ M(\lambda_{\kappa}, \ell)}{\lambda_{\kappa}}  - \E{\frac{M(\Lambda, L)}{\Lambda} \bigg\vert \Lambda=\lambda_{\kappa}} \right\}^2 \int \pi_\ell \left( \lambda_{\kappa}\theta,  \psi \right) g(\theta){\rm d}\theta
  $.
  \item[iii.] $ \E{\frac{M(\Lambda, L)}{\Lambda}  \bigg\vert \Lambda}  = \sum_{\ell =1}^{z}  \E{\frac{M(\Lambda, L)}{\Lambda}  \bigg\vert \Lambda, L=\ell} \int \pi_{\ell} \left( \Lambda\theta,  \psi \right) g(\theta){\rm d}\theta$.
  \item[iv.] $\E{\frac{M(\Lambda, L)}{\Lambda}  } = \sum\limits_{\kappa\in\mathcal{K}}  w_{\kappa} \E{\frac{M(\Lambda, L)}{\Lambda}  \big\vert \Lambda=\lambda_{\kappa}}$.
\end{enumerate}
\end{lemma}
\begin{proof}
  We omit the proof for this since the results are self-explanatory.
\end{proof}

\begin{remark}\label{rem.1}
The FIX in \eqref{eq.17} can be compared to the fairness measure defined in
  \citet{Chong}.
  Specifically, \citet{Chong} defined the fairness measure as a variation in the conditional expectation of the {\it a priori} rate as follows:
 \begin{equation}\label{eq.28}
  \frac{\Var{\E{\Lambda\big\vert L}}}{\Var{\Lambda}}.
  \end{equation}
  Note that while the measure in \eqref{eq.28} can be used to measure the fairness of the BMS in general, it cannot distinguish between two BMSs having distinct BM relativity tables. Since we are interested in the comparison of the BMSs with various BM relativity tables, the fairness measure in \eqref{eq.28} is not suitable for our purpose.
\end{remark}

\section{Table}

\begin{table}[h!t!]
\centering
\caption{Observable policy characteristics used as covariates} \label{tab.x}
\begin{tabular}{l|l r r r r r r r }
\hline
Categorical & \multirow{2}{*}{Description} &&  \multicolumn{3}{c}{\multirow{2}{*}{Proportions}} \\
variables &  &  &  &   \\
\hline
Entity type   & Type of local government entity    \\
		& \quad\quad\quad\quad\quad\quad Miscellaneous  	&& \multicolumn{3}{c}{5.03$\%$} \\
		& \quad\quad\quad\quad\quad\quad City			&& \multicolumn{3}{c}{9.66$\%$} \\
		& \quad\quad\quad\quad\quad\quad County			&& \multicolumn{3}{c}{11.47$\%$} \\
		& \quad\quad\quad\quad\quad\quad School			&& \multicolumn{3}{c}{36.42$\%$} \\
		& \quad\quad\quad\quad\quad\quad Town			&& \multicolumn{3}{c}{16.90$\%$} \\
		& \quad\quad\quad\quad\quad\quad Village 			&& \multicolumn{3}{c}{20.52$\%$} \\
\hline
Coverage & Collision coverage amount for old and new vehicles.\\
		& \quad\quad\quad\quad\quad\quad Coverage $\in (0,\quad\,\, 0.14] = 1 $   && \multicolumn{3}{c}{33.40$\%$} \\
		& \quad\quad\quad\quad\quad\quad Coverage $\in (0.14, \, 0.74] = 2 $	&& \multicolumn{3}{c}{33.20$\%$} \\
		& \quad\quad\quad\quad\quad\quad Coverage $\in (0.74,\quad \infty] = 3$	&& \multicolumn{3}{c}{33.40$\%$} \\
\hline
\end{tabular}
\end{table}

\end{document}